\documentclass[12pt]{article}
\usepackage{amsmath}
\usepackage{graphicx}
\usepackage{natbib}
\usepackage{url} 

\newcommand{\blind}{0}

\usepackage{amssymb, amsmath, amsthm, amsbsy, bm, mathrsfs}
	\allowdisplaybreaks[4]
	
	\newtheorem{Lemma}{Lemma}
	\newtheorem{Proposition}{Proposition}
	\newtheorem{Remark}{Remark}

\usepackage[hidelinks]{hyperref}
\usepackage{algorithm, algpseudocode}
\usepackage{textcomp}
\usepackage{caption, subcaption}
\usepackage{enumitem}
	\setenumerate[1]{itemindent=0pt, itemsep=0pt, partopsep=0pt, parsep=\parskip, topsep=0pt}
\usepackage{color}
\usepackage{array, multirow}

\newcommand{\aalpha}{\boldsymbol{\alpha}}

\newcommand{\dd}{{\rm d}}

\newcommand{\SSigma}{\boldsymbol{\Sigma}}

\DeclareMathOperator*{\argmin}{arg\,min\,}
\DeclareMathOperator*{\cov}{cov}

\DeclareMathOperator*{\E}{E}

\DeclareMathOperator*{\var}{var}

\begin{document}

\def\spacingset#1{\renewcommand{\baselinestretch}%
{#1}\small\normalsize} \spacingset{1}


\if0\blind
{
  \title{\bf Fast implementation of partial least squares for function-on-function regression}
  \author{Zhiyang Zhou\thanks{
    The author gratefully acknowledge the Natural Sciences and Engineering Council of Canada (NSERC) for financial supports.
    } \hspace{.2cm}\\
    Department of Preventive Medicine\\Northwestern University Feinberg School of Medicine}
  \maketitle
} \fi

\if1\blind
{
  \bigskip
  \bigskip
  \bigskip
  \begin{center}
    {\LARGE\bf Fast implementation of partial least squares for function-on-function regression}
\end{center}
  \medskip
} \fi

\bigskip
\begin{abstract}
	People employ the function-on-function regression
	to model the relationship between two random curves. 
	Fitting this model,
	widely used strategies include algorithms
	falling into the framework of functional partial least squares 
	(typically requiring iterative eigen-decomposition).
	Here we introduce a route of functional partial least squares 
	based upon Krylov subspaces.
	It can be expressed in two forms equivalent to each other (in exact arithmetic):
	one is non-iterative with explicit forms of estimators and predictions,
	facilitating the theoretical derivation 
	and potential extensions (to more complex models);
	the other one stabilizes numerical outputs. 
	The consistence of estimators and predictions is established under regularity conditions.
	Our proposal is highlighted as it is less computationally involved.
	Meanwhile,
	it is competitive in terms of both estimation and prediction accuracy.
\end{abstract}

\noindent%
\vfill

\newpage
\spacingset{1.2} 
\section{Introduction}\label{sec:intro}

Sometimes one would like to model the relationship between two stochastic curves.
To exemplify this type of interest,
two instances are listed as below.
\begin{description}
    \item[Diffusion tensor imaging (DTI) data] 
    	\citep[dataset \texttt{DTI} in \texttt{R} package \texttt{refund},]
    		[with missing values imputed through local polynomial regrssion]{R-refund}.
    	DTI is powerful
    	for characterizing microstructural changes for neuropathology
    	\citep{AlexanderLeeLazarField2007}.
        One of widely used DTI measures is the fractional anisotropy (FA).
        Along a tract of interest in the brain,
        FA values form
        an FA tract profile.
        Originally collected at the Johns Hopkins University and Kennedy-Krieger Institute,
        382 pairs of FA tract profiles for corpus callosum (CCA) and right corticospinal tract (RCST)
        are included in dataset \texttt{DTI} in \texttt{R} package \texttt{refund} \citep{R-refund}.
    	There are already interests on associations between 
    	FA trajectories for CCA and RCST;
    	see, e.g., \cite{IvanescuStaicuScheiplGreven2015}.
    \item[Boys' gait (BG) data] \citep[dataset \texttt{gait} in \texttt{R} package \texttt{fda},][]{R-fda}.
    	This dataset records hip and knee angles in degrees for 39 walking boys.
    	For each individual,
    	through a 20-point movement cycle,
    	these angles form two curves.
    	Then BG may be partially reflected by
    	the relationship between hip and knee curves.
\end{description}

As a fundamental model in the functional data analysis, 
the function-on-function regression 
(FoFR, first proposed by \citealp{RamsayDalzell1991}) 
may be helpful to the these scientific explorations.
Let $X=X(s)$ and $Y=Y(t)$
be two $L_2$-processes defined, respectively, 
on closed intervals $\mathbb I_X, \mathbb I_Y\subset\mathbb{R}$.
FoFR is formulated as
$$
	Y(t)=\mu_Y(t)+\int_{\mathbb I_X}\{X(s)-\mu_X(s)\}\beta^*(s,t)\dd s+\varepsilon(t),
$$
where $\beta^* \in L_2(\mathbb I_X\times \mathbb I_Y)$ is the target unknown parameter function
and $\mu_X(s)$ (resp. $\mu_Y(t)$) denotes $\E\{X(s)\}$ (resp. $\E\{Y(t)\}$).
Zero-mean Gaussian process $\varepsilon(t)$ has a covariance function 
continuous on $\mathbb I_Y\times \mathbb I_Y$
and is uncorrelated with $X(s)$,
i.e., $\E\{X(s), \varepsilon(t)\}=0$ for all $(s,t)\in\mathbb I_X\times\mathbb I_Y$.
This model becomes
$$
	Y(t)=\mu_Y(t)+\mathcal{L}_X(\beta^*)(t)+\varepsilon(t),
$$
defining a random integral operator 
$\mathcal{L}_X:L_2(\mathbb I_X\times \mathbb I_Y)\to L_2(\mathbb I_Y)$ such that,
for each $f\in L_2(\mathbb I_X\times \mathbb I_Y)$,
$$
	\mathcal{L}_X(f)(\cdot) = \int_{\mathbb I_X}\{X(s)-\mu_X(s)\}f(s,\cdot)\dd s.
$$
Write
$r_{XX}=r_{XX}(s,t)=\cov\{X(s),X(t)\}$,
and $r_{YY}=r_{YY}(s,t)=\cov\{Y(s),Y(t)\}$,
continuous respectively on $\mathbb I_X\times\mathbb I_X$ and $\mathbb I_Y\times\mathbb I_Y$.
Also,
we have continuous $r_{XY}=r_{XY}(s,t)=\cov\{X(s),Y(t)\}$, $(s,t)\in \mathbb I_X\times \mathbb I_Y$.
Correspondingly,
a linear integral operator $R_{XX}:L_2(\mathbb I_X)\to L_2(\mathbb I_X)$ is given by,
for each $f\in L_2(\mathbb I_X)$,
$R_{XX}(f)(\cdot) = \int_{\mathbb I_X}r_{XX}(\cdot,t)f(t)\dd t$.
One more operator $R_{YY}: L_2(\mathbb I_Y)\to L_2(\mathbb I_Y)$ is defined in complete analogy to $R_{XX}$.
Let $(\lambda_{j,X}, \phi_{j,X})$ (resp. $(\lambda_{j,Y}, \phi_{j,Y})$) 
be the two-tuple consisting of the $j$th leading eigenvalue and eigenfunction of 
$R_{XX}$ (resp. $R_{YY}$).
It is standard for functional data analysis to assume that 
$\sum_{j=1}^\infty\lambda_{j,X}<\infty$ and $\sum_{j=1}^\infty\lambda_{j,Y}<\infty$,
with positive $\lambda_{j,X}$ and $\lambda_{j,Y}$.
Ensuring the identifiability of $\beta^*$,
condition \ref{cond:identifiable} (assumed by, e.g., \citealp{HeMullerWangYang2010, YaoMullerWang2005b})
derives a closed-form of $\beta^*$, i.e.,
for each $(s,t)\in \mathbb I_X\times \mathbb I_Y$,
\begin{equation}\label{eq:beta*}
	\beta^*(s,t)
	= \Gamma_{XX}^{-1}(r_{XY})(s,t)
	= \sum_{j, j'=1}^\infty\frac
		{\int_{\mathbb I_Y}\int_{\mathbb I_X}\phi_{j,X}(s)r_{XY}(s,t)\phi_{j',Y}(t)\dd s\dd t}
		{\lambda_{j,X}}
		\phi_{j,X}(s)\phi_{j',Y}(t),
\end{equation}
where $\|\cdot\|_2$ denotes the $L_2$-norm
(and is abused for all the $L_2$ spaces involved hereafter);
and $\Gamma_{XX}:L_2(\mathbb I_X\times \mathbb I_Y)\to L_2(\mathbb I_X\times \mathbb I_Y)$
is a linear integral operator defined as,
for each $f\in L_2(\mathbb I_X\times \mathbb I_Y)$,
$$
	\Gamma_{XX}(f)(s, t) = 
		\int_{\mathbb I_X}r_{XX}(s,s')f(s', t)\dd s',\quad (s,t)\in \mathbb I_X\times \mathbb I_Y.
$$

Excellent contributions have been made to the investigation of FoFR.
In general,
due to the intrinsically infinite dimension,
people have to consider 
an approximation to $\beta^*$ within certain subspaces of $L_2(\mathbb I_X\times\mathbb I_Y)$.
Traditionally,
these subspaces are constructed from pre-determined functions, 
e.g., splines and Fourier basis functions.
But a more prevailing option may be data-driven: 
the functional principal component regression (FPCR)
drops the tail of the series on the farthest right-hand side of \eqref{eq:beta*}
and approximates $\beta^*$ by its orthogonal projection to
${\rm span}\{f_{jj'}\in L_2(\mathbb I_X\times \mathbb I_Y)\mid 
    f_{jj'}(s,t)=\phi_{j,X}(s)\phi_{j',Y}(t), 1\leq j\leq p, 1\leq j'\leq p'\}$,
with $p$ and $p'$ chosen by cross-validation 
and ${\rm span}(\cdot)$ denoting the linear space spanned by elements inside the braces;
specifically,
FPCR approximates $\beta^*$ by
\begin{equation}\label{eq:beta.p.q.fpca}
	\beta_{p,p',\rm FPCR}(s,t)
	= \sum_{j=1}^p\sum_{j'=1}^{p'}
		\frac{\int_{\mathbb I_Y}\int_{\mathbb I_X}\phi_{j,X}(v)r_{XY}(v,w)\phi_{j',Y}(w)\dd v\dd w}{\lambda_{j,X}}\phi_{j,X}(s)\phi_{j',Y}(t).
\end{equation}
Accompanied with a penalized estimation,
\citet{Lian2015} and \citet{SunDuWangMa2018} 
limited their discussions on coefficient estimators to reproducing kernel Hilbert spaces.
The Tikhonov (viz. ridge-type) regularization in \citet{BenatiaCarrascoFlorens2017}
yields a remedy for ill-posed $\beta^*$ when not all $\lambda_{j,X}$ are non-zero.
Distinct from these works,
our consideration is based on 
a subspace of $L_2(\mathbb I_X\times \mathbb I_Y)$ named after (Alexei) Krylov, viz.
\begin{equation}\label{eq:krylov}
	{\rm KS}_p(\Gamma_{XX}, \beta^*)
	={\rm span}\{\Gamma_{XX}^j(\beta^*)\mid 1\leq j\leq p\},
\end{equation}
where $\Gamma_{XX}^0$ is indeed the identity operator $I$, 
while $\Gamma_{XX}^j:L_2(\mathbb I_X\times \mathbb I_Y)\to L_2(\mathbb I_X\times \mathbb I_Y)$, $j\geq 1$,
is defined recursively as,
for each $f\in L_2(\mathbb I_X\times \mathbb I_Y)$ and each $(s,t)\in \mathbb I_X\times \mathbb I_Y$,
\begin{align*}
	\Gamma_{XX}^j(f)(s,t)
	=&\ (\Gamma_{XX}\circ\Gamma_{XX}^{j-1})(f)(s,t)
	\\
	=&\ \Gamma_{XX}\{\Gamma_{XX}^{i-1}(f)\}(s,t)
	\\
	=&\ \int_{\mathbb I_X}r_{XX}(s,w)\{\Gamma_{XX}^{i-1}(f)(w,t)\}\dd w.
\end{align*}
Noting that $\Gamma_{XX}^j(\beta^*)=\Gamma_{XX}^{j-1}(r_{XY})$ for all $j\in\mathbb{Z}^+$,
the ($p$-dimensional) Krylov subspace at \eqref{eq:krylov} incorporates both $X$ and $Y$ 
and hence overcomes the unsupervision of truncated eigenspace used for FPCR.

The subspace at \eqref{eq:krylov} is a generalization of \citet[(3.4)]{DelaigleHall2012b},
expanding as well the Krylov subspace method
for the (multivariate) partial least squares (PLS).
In the multivariate context,
PLS is a terminology shared by a series of algorithms
yielding supervised (i.e., related-to-response) basis functions;
\citet[Section~2.2]{Bissett2015} briefed several well-known examples of them,
including the nonlinear iterative PLS 
(NIPALS, \citealp{Wold1975})
and the statistically inspired modification of PLS 
(SIMPLS, \citealp{deJong1993}). 
For single-vector-response,
these two lead to outputs identical to that from the Krylov subspace method;
but they are known to yield different results when the response is of more than one vectors;
see \citet[Section~7.2]{CookForzani2019}.
Likewise,
their respective functional counterparts are equivalent to each other for scalar-response
but become diverse again for FoFR.
We refer readers to \citet{BeyaztasShang2020}
for a straightforward extension of NIPALS and SIMPLS for FoFR.
Shooting at the same model,
SigComp \citep{LuoQi2017} embeds penalties into NIPALS.
It is \autoref{prop:conv.beta.p} that drives us to pick up the Krylov subspace method as our route.
\begin{Proposition}\label{prop:conv.beta.p}
	Under \ref{cond:identifiable},
	true parameter $\beta^*\in
	\overline{{\rm KS}_\infty(\Gamma_{XX}, \beta^*)}
	=\overline{{\rm span}\{\Gamma_{XX}^j(\beta^*)\mid j\geq 1\}}$,
	with the overline representing the closure.
\end{Proposition}
\begin{Remark}
	It is worth noting that
	\autoref{prop:conv.beta.p} is not a corollary of \citet[Theorem~3.2]{DelaigleHall2012b};
	the latter one merely implies an identity weaker than \autoref{prop:conv.beta.p}: 
	fixing arbitrary $t_0\in\mathbb{I}_Y$,
	univariate function 
	$\beta^*(\cdot, t_0)\in
	\overline{{\rm span}\{\Gamma_{XX}^j(\beta^*)(\cdot, t_0)\mid j\geq 1\}}$.
\end{Remark}

As an extension of the alternative PLS \citep[APLS,][designed for the scalar-on-function regression]{DelaigleHall2012b},
our proposal is abbreviated as fAPLS,
with letter ``f'' emphasizing its application to FoFR. 
The remaining portion of this paper is organized as below.
\autoref{sec:method} details two equivalent expressions of fAPLS estimators,
facilitating the empirical implementation and theoretical derivation, respectively.
In \autoref{sec:numerical} 
fAPLS is compared with competitors 
and is illustrated as a time-saving option.
The framework of fAPLS is potential to be extended to more complex settings,
e.g., correlated subjects and non-linear modelling;
we include three promising directions in \autoref{sec:discussion}.
More assumptions and proofs are relegated to Appendix for conciseness.

\section{Method}\label{sec:method}

We propose to project $\beta^*$ to \eqref{eq:krylov} and 
to utilize the least squares solution
\begin{equation}\label{eq:beta.p.fapls}
	\beta_{p,\rm fAPLS}
	= \argmin_{\beta\in {\rm KS}_p(\Gamma_{XX}, \beta^*)}\E\|Y-\mu_Y-\mathcal{L}_X(\beta)\|_2^2
	= [\Gamma_{XX}(\beta^*),\ldots,\Gamma_{XX}^p(\beta^*)]\bm H_p^{-1}\bm\alpha_p,
\end{equation}
where $\bm H_p=[h_{jj'}]_{1\leq j,j'\leq p}$ and $\bm\alpha_p=[\alpha_1,\ldots,\alpha_p]^\top$ 
denote $p\times p$ and $p\times 1$ matrices,
respectively, with
\begin{align}
	h_{jj'} 
		=&\ \int_{\mathbb I_Y}\left\{
			\int_{\mathbb I_X}\int_{\mathbb I_X}r_{XX}(s, w)\Gamma_{XX}^j(\beta^*)(s, t)\Gamma_{XX}^{j'}(\beta^*)(w, t)\dd s\dd w
		\right\}\dd t\notag
	\\
		=& \int_{\mathbb I_Y}\int_{\mathbb I_X}\Gamma_{XX}^j(\beta^*)(s, t)\Gamma_{XX}^{j'+1}(\beta^*)(s, t)\dd s\dd t,\label{eq:h.ij}
	\\
	\alpha_i 
		=&\ \int_{\mathbb I_Y}\left\{
			\int_{\mathbb I_X}\int_{\mathbb I_X}r_{XX}(s, w)\Gamma_{XX}^j(\beta^*)(s, t)\beta^*(w, t)\dd s\dd w
		\right\}\dd t\notag
	\\
		=&\ \int_{\mathbb I_Y}\int_{\mathbb I_X}\Gamma_{XX}(\beta^*)(s, t)\Gamma_{XX}^j(\beta^*)(s, t)\dd s\dd t.\notag
\end{align}
\autoref{prop:conv.beta.p} justifies \eqref{eq:beta.p.fapls} 
by entailing that $\lim_{p\to\infty}\|\beta_{p,\rm fAPLS}-\beta^*\|_2=0$,
which is crucial to the consistency of our estimators delivered later.

Suppose $n$ two-tuples $(X_i, Y_i)$, $1\leq i\leq n$,
are all independent realizations of $(X, Y)$.
Nobody is aware of the analytical expressions of these trajectories.
So it is impossible to compute corresponding integrals exactly.
Nevertheless,
numerical tools like quadrature rules are available and satisfactory,
as long as observed points at each curve are sufficiently dense.
Errors are introduced in these approximations.
Though they are bounded,
it is inevitable to assume smoothness of original trajectories;
see, e.g., \citet{Tasaki2009} for the trapezoidal rule.
To fulfill the requirement on smoothness, 
interpolations, e.g., various splines, are often involved;
refer to, e.g., \citet{Xiao2019} for theoretical results on certain penalized splines.
For convenience,
we assume curves to be observed densely enough and
abuse integral signs for corresponding empirical approximations.

It is natural to estimate $r_{XX}(s,s')$ 
and $r_{XY}(s,t)$ ($=\Gamma_{XX}(\beta^*)(s,t)$), 
$(s,s',t)\in\mathbb{I}_X\times\mathbb{I}_X\times\mathbb{I}_Y$,
respectively, by
\begin{align}
    \hat r_{XX}(s,s') &= \frac{1}{n}\sum_{i=1}^{n}X_i^{\rm cent}(s)X_i^{\rm cent}(s')
    \label{eq:r.xx.hat}\\
	\hat r_{XY}(s,t) &= \frac{1}{n}\sum_{i=1}^{n}X_i^{\rm cent}(s)Y_i^{\rm cent}(t)
		(=\widehat\Gamma_{XX}(\beta^*)(s,t))
	\label{eq:r.xy.hat}
\end{align}
in which $X_i^{\rm cent}=X_i-\bar X$ and $Y_i^{\rm cent}=Y_i-\bar Y$,
with $\bar X=n^{-1}\sum_{i=1}^{n}X_i$ and $\bar Y=n^{-1}\sum_{i=1}^{n}Y_i$.
Given $\widehat\Gamma_{XX}^j(\beta^*)$,
one can estimate $\Gamma_{XX}^{j+1}(\beta^*)(s,t)$ by
\begin{equation}\label{eq:gamma.xx.beta.i}
	\widehat\Gamma_{XX}^{i+1}(\beta^*)(s,t)
	=\int_{\mathbb I_X}\hat r_{XX}(s,s')\widehat\Gamma_{XX}^j(\beta^*)(s',t)\dd s'.
\end{equation}
Plugging \eqref{eq:r.xx.hat}, \eqref{eq:r.xy.hat} and \eqref{eq:gamma.xx.beta.i} 
all into \eqref{eq:beta.p.fapls},
an estimator for both $\beta_{p,\rm fAPLS}$ and $\beta^*$ comes:
\begin{equation}\label{eq:beta.p.fapls.hat}
	\hat\beta_{p,\rm fAPLS}
	= [\widehat\Gamma_{XX}(\beta^*),\ldots,\widehat\Gamma_{XX}^p(\beta^*)]\widehat{\bm H}_p^{-1}\widehat{\bm\alpha}_p,
\end{equation}
where $\widehat{\bm H}_p=[\hat h_{jj'}]_{1\leq j,j'\leq p}$ and $\widehat{\bm\alpha}_p=[\hat\alpha_1,\ldots,\hat\alpha_p]^\top$
are respectively consisting of
\begin{align}
	\hat h_{jj'} &= \int_{\mathbb I_Y}\int_{\mathbb I_X}
			\widehat\Gamma_{XX}^j(\beta^*)(s, t)\widehat\Gamma_{XX}^{j'+1}(\beta^*)(s, t)
		\dd s\dd t,\label{eq:h.ij.hat}
	\\
	\hat\alpha_j &= \int_{\mathbb I_Y}\int_{\mathbb I_X}
			\widehat\Gamma_{XX}(\beta^*)(s, t)\widehat\Gamma_{XX}^j(\beta^*)(s, t)
		\dd s\dd t.\notag
\end{align}
Finally,
given trajectory $X_0\sim X$ and $t\in\mathbb I_Y$,
\begin{equation}\label{eq:g}
	g(X_0)(t) = \E\{Y(t)\mid X=X_0\} = \mu_Y(t) + \mathcal L_{X_0}(\beta^*)(t)
\end{equation}
is predicted by
\begin{equation}\label{eq:g.p.fapls.hat}
	\hat g_{p,\rm fAPLS}(X_0)(t) = \bar Y(t) + \int_{\mathbb I_X}X_0^{\rm cent}(s)\hat\beta_{p,\rm fAPLS}(s,t)\dd s.
\end{equation}

$\widehat{\bm H}$ at \eqref{eq:beta.p.fapls.hat} is always invertible if we were able to work in exact arithmetic.
But it is not the case for finite precision arithmetic:
as $p$ increases,
the linear system from $\widehat\Gamma_{XX}(\beta^*),\ldots,\widehat\Gamma_{XX}^p(\beta^*)$ may be close to singular.
To overcome this numerical difficulty,
as suggested by \citet[Section~4.2]{DelaigleHall2012b},
we orthonormalize $\widehat\Gamma_{XX}(\beta^*),\ldots,\widehat\Gamma_{XX}^p(\beta^*)$ (with respect to $\hat r_{XX}$)
into $\hat\psi_1,\ldots,\hat\psi_p$ (see \autoref{alg:mgs} or \citealt[pp.~102]{Lange2010})
and reformulate the optimization problem at \eqref{eq:beta.p.fapls} into the empirical version:
\begin{equation}\label{eq:max.stable}
	\max_{[c_1,\ldots,c_p]^\top\in\mathbb R^p}
	\frac{1}{n}\sum_{i=1}^{n}\int_{\mathbb I_Y}\left\{
		Y_i(t)-\bar Y(t)-\sum_{j=1}^p c_j\int_{\mathbb I_X}X_i^{\rm cent}(s)\hat\psi_j(s,t)\dd s
	\right\}^2\dd t.
\end{equation}
We then reach a numerically stabilized estimator for $\beta^*$:
\begin{equation}\label{eq:beta.p.fapls.tilde}
	\tilde\beta_{p,\rm fAPLS}
	= [\hat\psi_1,\ldots,\hat\psi_p][\hat\gamma_1,\ldots,\hat\gamma_p]^\top
	=\sum_{j=1}^p\hat\gamma_j\hat\psi_j,
\end{equation}
where $[\hat\gamma_1,\ldots,\hat\gamma_p]^\top$ is the maximizer of \eqref{eq:max.stable},
with 
$$
	\hat\gamma_j=\int_{\mathbb I_Y}\int_{\mathbb I_X}\hat r_{XY}(s, t)\hat\psi_j(s, t)\dd s\dd t.
$$
A prediction for $g(X_0)$ at \eqref{eq:g}, 
alternative to $\hat g_{p,\rm fAPLS}(X_0)$ at \eqref{eq:g.p.fapls.hat}, 
is thus given by
\begin{equation}\label{eq:g.p.fapls.tilde}
	\tilde g_{p,\rm fAPLS}(X_0)(t) 
	= \bar Y(t) + \int_{\mathbb I_X}X_0^{\rm cent}(s)\tilde\beta_{p,\rm fAPLS}(s,t)\dd s.
\end{equation}

It is worth emphasizing that,
in exact arithmetic,
$\hat\beta_{p,\rm fAPLS}$ at \eqref{eq:beta.p.fapls.hat}
(resp. $\hat g_{p,\rm fAPLS}$ at \eqref{eq:g.p.fapls.hat}) 
is identical to $\tilde\beta_{p,\rm fAPLS}$ at \eqref{eq:beta.p.fapls.tilde}
(resp. $\tilde g_{p,\rm fAPLS}$ at \eqref{eq:g.p.fapls.tilde}),
because $\{\widehat\Gamma_{XX}^j(\beta^*)\mid 1\leq j\leq p\}$
and $\{\hat\psi_j\mid 1\leq j\leq p\}$
literally span the same space.
Nevertheless,
in practice $\tilde\beta_{p,\rm fAPLS}$ and $\tilde g_{p,\rm fAPLS}$ stand out 
due to their numerical stability for finite precision arithmetic,
whereas the more explicit expressions of $\hat\beta_{p,\rm fAPLS}$ and $\hat g_{p,\rm fAPLS}$
make themselves preferred in theoretical derivations.

\begin{algorithm}[t!]
	\caption{Modified Gram-Schmidt orthonormalization with respect to $\hat r_{XX}$}
	\label{alg:mgs}
	\begin{algorithmic}[]
		\For {$j$ in $1,\ldots,p$}
		    \State $\hat\psi_j^{[1]}\gets\widehat\Gamma_{XX}^j(\beta^*)$.
		    \If {$j\geq 2$}
    		    \For {$j'$ in $1,\ldots,j-1$}
    		        \State $
    		        	\hat\psi_j^{[j'+1]}\gets
    		            	\hat\psi_j^{[j']}-
    		                \left\{
    		                	\int_{\mathbb I_Y}\int_{\mathbb I_X}\int_{\mathbb I_X}
    		                		\hat r_{XX}(s,s')\hat\psi_j^{[j']}(s,t)\hat\psi_{j'}(s',t)
    		                	\dd s\dd s'\dd t
    		       			\right\}\hat\psi_{j'}
    		       		$.
    		    \EndFor
    		\EndIf
		    \State $\hat\psi_j\gets
					\left\{
	                	\int_{\mathbb I_Y}\int_{\mathbb I_X}\int_{\mathbb I_X}
	                		\hat r_{XX}(s,s')\hat\psi_j^{[j]}(s,t)\hat\psi_j^{[j]}(s',t)
	                	\dd s\dd s'\dd t
	       			\right\}^{-1/2}\hat\psi_j^{[j]}
		    	$.
		\EndFor
	\end{algorithmic}
\end{algorithm}

We have one hyper-parameter to tune.
Using five-fold cross-validation,
$p$ is chosen as the minimizer of
$$
    {\rm CV}(p)=\frac{1}{5}\sum_{k=1}^{5}
    	\frac{\sum_{i\in I_k}\|Y_i-\tilde g_{p,\rm fAPLS}^{(-k)}(X_i)\|_2^2}
        	{\sum_{i\in I_k}\|Y_i-\sum_{i\in I_{\rm test}\setminus I_k}Y_i/(\# I_{\rm test}-\# I_k)\|_2^2},
$$
where $\{I_1,\ldots, I_5\}$ is a partition of index set for testing,
say $I_{\rm test}$;
where \# represents the cardinality;
where $\tilde g_{p,\rm fAPLS}^{(-k)}(X_i)$ predicts $g(X_i)$ 
and is constructed from data points corresponding to $I_{\rm test}\setminus I_k$.
Define the fraction of variance explained (FVE) as 
${\rm FVE}(p)=\sum_{j=1}^p\lambda_{j,X}/\sum_{j=1}^\infty\lambda_{j,X}$; then
the search for $p$ is limited within $[1,p_{\max}]$,
where $p_{\max}$ is set to be the smallest integer such that 
${\rm FVE}(p_{\max})$ exceeds a pre-determined close-to-one threshold, e.g., 99\%. 
This FVE criterion is commonly used in truncating Karhunen-Lo\`eve series,
e.g., FPCR.
Since FPLS algorithms are typically more parsimonious
than FPCR in terms of number of basis functions,
$p_{\max}$ formed in this way tends to be reasonable.

\subsection{Asymptotic properties}

Under regularity conditions,
\autoref{prop:conv.beta.p.hat} 
(resp. \autoref{prop:conv.g.p.hat})
verifies the consistency in $L_2$ and/or supremum metric
(in probability)
of $\hat\beta_{p,\rm fAPLS}$
(resp. $\hat g_{p,\rm fAPLS}(X_0)$).
In these results,
we allow $p$ to diverge as a function of $n$,
but its rate is capped to be at most $O(\sqrt{n})$ if $\|r_{XX}\|_2<1$
and even slower otherwise.
More discussion of the technical assumptions may be found at the beginning of Appendix.

\begin{Proposition}\label{prop:conv.beta.p.hat}
    Holding \ref{cond:identifiable}--\ref{cond:diverge.p.add.L2},
    as $n$ diverges,
    $\|\hat{\beta}_{p,\rm fAPLS}-\beta^*\|_2=o_p(1)$.
    If upgrade \ref{cond:diverge.p.add.L2} to \ref{cond:diverge.p.add.sup},
    then the convergence becomes uniform, 
    i.e.,
    $\|\hat{\beta}_{p,\rm fAPLS}-\beta^*\|_\infty=o_p(1)$,
    with $\|\cdot\|_\infty$ denoting the supremum metric.
\end{Proposition}

\begin{Proposition}\label{prop:conv.g.p.hat}
    Given $X_0\sim X$, 
    conditions \ref{cond:identifiable}--\ref{cond:diverge.p.add.L2}
    suffice for the zero-convergence (in probability) of $\|\hat g_{p,\rm fAPLS}(X_0)-g(X_0)\|_2$
    (i.e., $\|\hat g_{p,\rm fAPLS}(X_0)-g(X_0)\|_2=o_p(1)$),
    while the uniform version
    (viz. $\|\hat g_{p,\rm fAPLS}(X_0)-g(X_0)\|_\infty=o_p(1)$)
    is entailed jointly by \ref{cond:identifiable}--\ref{cond:inf.norm} 
    and \ref{cond:diverge.p.add.sup}--\ref{cond:moment.add.sup}.
\end{Proposition}

\section{Numerical study}\label{sec:numerical}

Our proposal fAPLS was compared with competitors 
in terms of 
the relative integrated squared estimation error (ReISEE)
and/or
relative integrated squared prediction error (ReISPE):
\begin{align*}
	{\rm ReISEE}&=
	\frac{
	        \|\beta^*-\hat\beta\|_2^2
	    }{
	        \|\beta^*\|_2^2
	    },
	\\
	{\rm ReISPE}&=
	\frac{
	        \sum_{i\in I_{\rm test}}\|Y_i-\hat Y_i\|_2^2
	    }{
	        \sum_{i\in I_{\rm test}}\|Y_i-\sum_{i\in I_{\rm train}}Y_i/\#I_{\rm train}\|_2^2
	    },
\end{align*}
where $\hat\beta$ estimates $\beta$ and $\hat{Y}_i$ predicts $Y_i$, $1\leq i\leq n$;
where \# represents the cardinality,
and $I_{\rm train}$ is the index set for training.
We reported the averages and standard deviations of ReISEEs and ReISPEs
in \autoref{tab:error} for all the numerical studies.
Subsequent comparisons involved other FPLS routes for FoFR,
including SigComp \citep{LuoQi2017} and
(functional) NIPALS and SIMPLS \citep{BeyaztasShang2020}.
These three routes seemed superior to quite a few competitors in literature.
We referred to their original source codes posted at
\texttt{R} package \texttt{FRegSigCom} \citep{R-FRegSigCom} and
GitHub (\url{https://github.com/hanshang/FPLSR}; accessed on \today),
respectively.
Code trunks for our implementation are already available too at GitHub
(\url{https://github.com/ZhiyangGeeZhou/fAPLS}; accessed on \today).

\subsection{Simulation}

In total we went through three simulation scenarios. 
They varied from each other 
on the settings of $\mu_Y$, $X$ and $\beta^*$ 
(as specified later)
but shared the identical setup for error term $\varepsilon=\varepsilon(t)$
which was a zero-mean Gaussian process with covariance function
$\E\{\varepsilon(t),\varepsilon(t')\}=\sigma_\varepsilon^2\rho^{|t-t'|}$,
$t,t'\in [0,1]$ ($=\mathbb{I}_X=\mathbb{I}_Y$ in simulation).
Given $\mu_Y$, $X$ and $\beta^*$,
parameters $\rho$ and $\sigma_\varepsilon^2$ determined the signal-noise-ratio (SNR),
viz. the ratio between 
$[\int_{\mathbb I_Y}\var\{\mathcal{L}_X(\beta^*)(t)\}\dd t]^{1/2}$
and $[\int_{\mathbb I_Y}\var\{\varepsilon(t)\}\dd t]^{1/2}$.
$\rho$ took either 0.1 (low autocorrelation) or 0.9 (high autocorrelation),
while two levels of $\sigma_\varepsilon^2$ were set up
so that SNR was moderate and fell between roughly 1 and 10;
see Tables \ref{tab:error} and \ref{tab:ancillary} for specific values of $\rho$ and $\sigma_\varepsilon^2$.
In each scenario,
we generated $n$ ($=300$) independent and identically distributed (i.i.d.)
pairs of trajectories 
(with 80\% kept for training and 20\% for testing).
Each curve was recorded 
at 101 equally spaced points in $[0,1]$,
specifically, $\{0, 1/101, \ldots, 100/101, 1\}$.
We repeated this procedure 50 times
for each combination of $\mu_Y$, $X$, $\beta^*$, $\rho$ and $\sigma_\varepsilon^2$.

\subsubsection{Simulation 1}\label{sec:simu.1}

Assume $\mu_Y=0$.
We took 100, 10 and 1
as the top three eigenvalues of $\Gamma_{XX}$,
whereas $\lambda_{j,X}=0$ for all $j\geq 4$.
Correspondingly,
the first three eigenfunctions of $\Gamma_{XX}$
were respectively set to be
(normalized) shifted Legendre polynomials of order 2 to 4
\citep[say $P_2$, $P_3$ and $P_4$; see][pp.~773--774]{Hochstrasser1972},
viz.
\begin{align*}
    \phi_{1,X}(s) &= P_2(s)= \sqrt{5}(6s^2 - 6s + 1), \\
    \phi_{2,X}(s) &= P_3(s)= \sqrt{7}(20s^3 - 30s^2 + 12s - 1), \\
    \phi_{3,X}(s) &= P_4(s)= 3(70s^4 - 140s^3 + 90s^2 - 20s + 1).
\end{align*}
As is known, 
they are of unit norm and mutually orthogonal on $[0,1]$.
The predictors and slope function were respectively given by
\begin{align*}
	X_i(s) &= \zeta_{i1}P_2(s) + \zeta_{i2}P_3(s) + \zeta_{i3}P_4(s),
	\\
	\beta^*(s,t) &= P_2(s)P_2(t)+P_3(s)P_3(t)+P_4(s)P_4(t),
\end{align*}
with $\zeta_{ij}$ independently distributed as $\mathcal{N}(0, \lambda_{j,X})$,
$j=1,\ldots,3$.

\nameref{sec:simu.1} was equipped with a true coefficient 
belonging to 
${\rm KS}_3(\Gamma_{XX}, \beta^*)$
and was in favor of our proposal.
As expected, 
fAPLS enjoyed lower estimation errors for this scenario;
see \autoref{tab:error}.
Nevertheless,
as for prediction errors,
the outputs from all the four methods were fairly comparable.
We speculated that
their extra estimation bias fell outside 
the range of $\Gamma_{XX}$
($=\{\Gamma_{XX}(f)\mid f\in L_2(\mathbb{I}_X\times\mathbb{I}_Y)\}$)
and hence impacted little on prediction.
ReISEEs of all methods
changed little with $\rho$ and $\sigma_\varepsilon^2$,
while their prediction accuracy was sensitive to $\sigma_\varepsilon^2$:
as $\sigma_\varepsilon^2$ became smaller,
prediction errors were all lowered.
Meanwhile,
four FPLS routes all chose around two components.
The biggest advantage of fAPLS was on the running time;
it ran faster than the other three in numerical studies;
see \autoref{tab:ancillary}.
This phenomenon was not surprising,
because, compared with others,
fAPLS involves no eigendecomposition and fewer tuning parameters.

\subsubsection{Simulation 2}\label{sec:simu.2}

Define two covariance functions as follows:
\begin{align*}
	\SSigma_1 &=\SSigma_1(s,s')=\exp\{-(10|s-s'|)^2\},
	\\
	\SSigma_2 &=\SSigma_2(s,s')=\{1+20|s-s'|+\frac{1}{3}(20|s-s'|)^2\}\exp(-20|s-s'|).
\end{align*}
Then generate $\zeta_1,\ldots,\zeta_7$ as i.i.d. realizations of 
the zero-mean Gaussian process with covariance function $\SSigma_2$.
Fixing $\zeta_1,\ldots,\zeta_7$ for this scenario,
we constructed
\begin{align*}
	\mu_Y(t) &= \zeta_1(t),
	\\
	\beta^*(s,t) &= \zeta_2(s)\zeta_3(t) + \zeta_4(s)\zeta_5(t) + \zeta_6(s)\zeta_7(t).
\end{align*}
Our setup is finished by sampling $X_i$, $1\leq i\leq 300$,
from the zero-mean Gaussian process with covariance function $\SSigma_1$.
This setting appeared too in \citet[Section~4.1.1]{LuoQi2017}.

The performance of four approaches was analogous to that in \nameref{sec:simu.1}:
fAPLS stood out in terms of estimation accuracy,
while prediction errors from all routes were pretty close.
Though the four methods shared the identical search scope for number of components,
models from fAPLS and SigComp were typically more parsimonious (viz. of fewer numbers of components)
than the remaining two.
Especially,
when there was more noise (viz. $\sigma_\varepsilon^2=80$) in \nameref{sec:simu.2},
fAPLS built up most concise models with little loss in estimation and prediction accuracy;
see \autoref{tab:ancillary}.

\subsubsection{Simulation 3}

We considered functional predictors and coefficient
functions similar to those in
\citet[Section~4.1]{GoldsmithBobCrainiceanuCaffoReich2011},
\citet[Section~4.1]{IvanescuStaicuScheiplGreven2015}
and \citet[Section~4.1.2]{LuoQi2017}:
\begin{align*}
	\mu_Y(t) &= 2\exp\{-(t-1)^2\},
	\\
	\beta^*(s,t) &= \sin(\pi s)\cos(2\pi t),
	\\
	X_i(s) &= \sum_{m=1}^{10}\frac{1}{m^2}\{\zeta_{i1m}\sin(m\pi s)+\zeta_{i2m}\cos(m\pi s)\},
\end{align*}
where $\zeta_{ijm}$, $1\leq i\leq 300$, $1\leq j\leq 2$, $1\leq m\leq 10$,
are all i.i.d. standard normal.
This was a scenario where SIMPLS worked generally better than others in estimation,
while fAPLS was competitive as long as 
either $\rho$ or $\sigma_\varepsilon^2$ was not at the low level. 
In terms of prediction,
all the methods were still of similar accuracy.
Number of components picked up by SigComp were in average about 0.5 fewer than those for fAPLS,
whereas the latter one ran significantly faster;
see \autoref{tab:ancillary}.

\subsection{Application}

Revisit the two datasets described in \autoref{sec:intro}.
For DTI (resp. BG) data,
we took CCA FA tract profiles (resp. hip angle curves) 
illustrated at \autoref{fig:cca} (resp. \autoref{fig:hip}) as predictors 
and RCST FA tract profiles (resp. knee angle curves) 
illustrated at \autoref{fig:rcst} (resp. \autoref{fig:knee}) as responses.
For each dataset,
repeat the following random split for 50 times:
take roughly 20\% of all subjects
for testing and the remaining for training.
We thus generated 50 ReISPEs for each dataset and each approach.

\begin{figure}[!t]
	\centering
	\begin{subfigure}{.49\textwidth}
		\centering
		\includegraphics[width=.7\textwidth, height=.25\textheight]
			{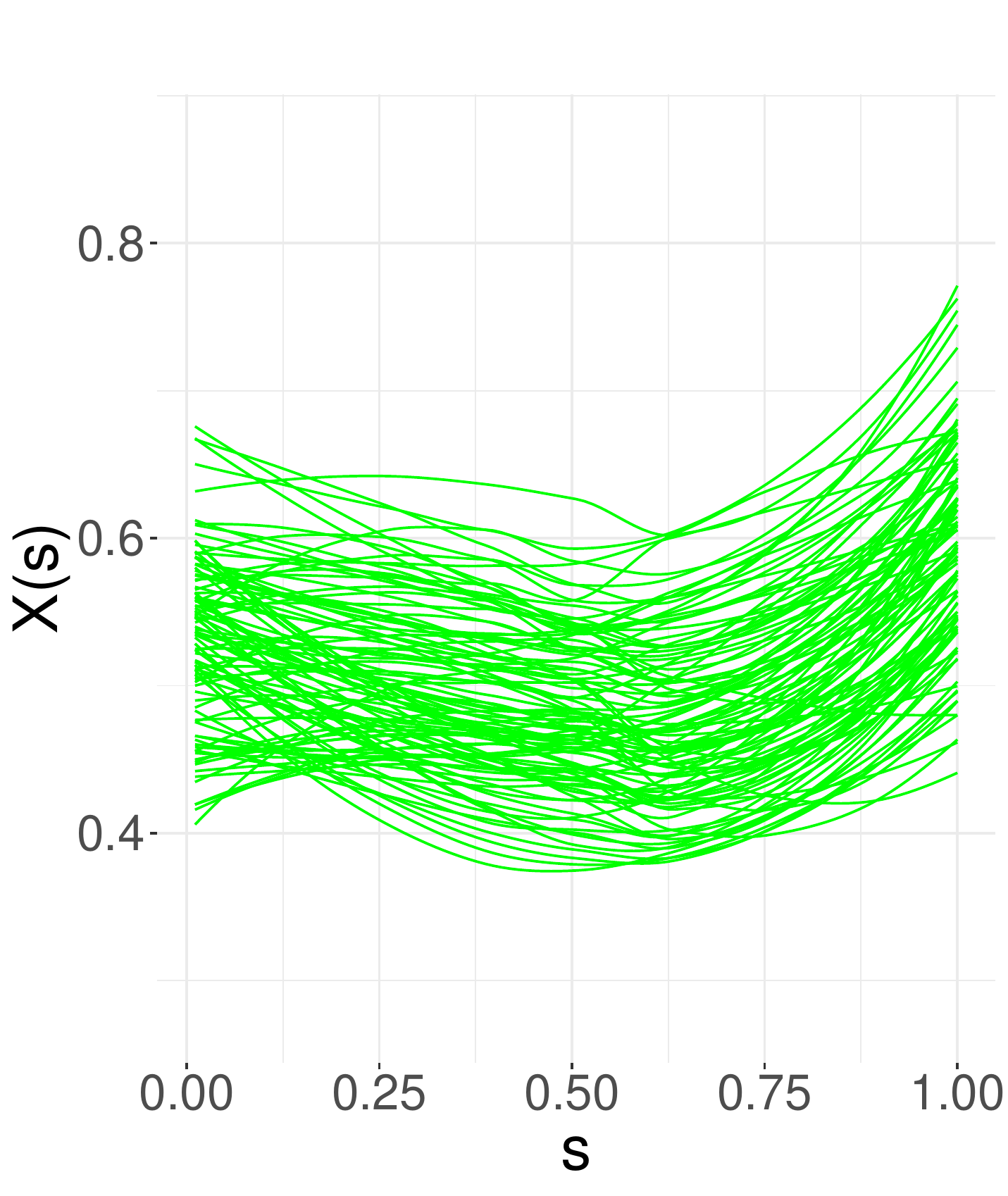}
		\caption{CCA FA tract profile ($X$)}
		\label{fig:cca}
	\end{subfigure}
	\begin{subfigure}{.49\textwidth}
		\centering
		\includegraphics[width=.7\textwidth, height=.25\textheight]
			{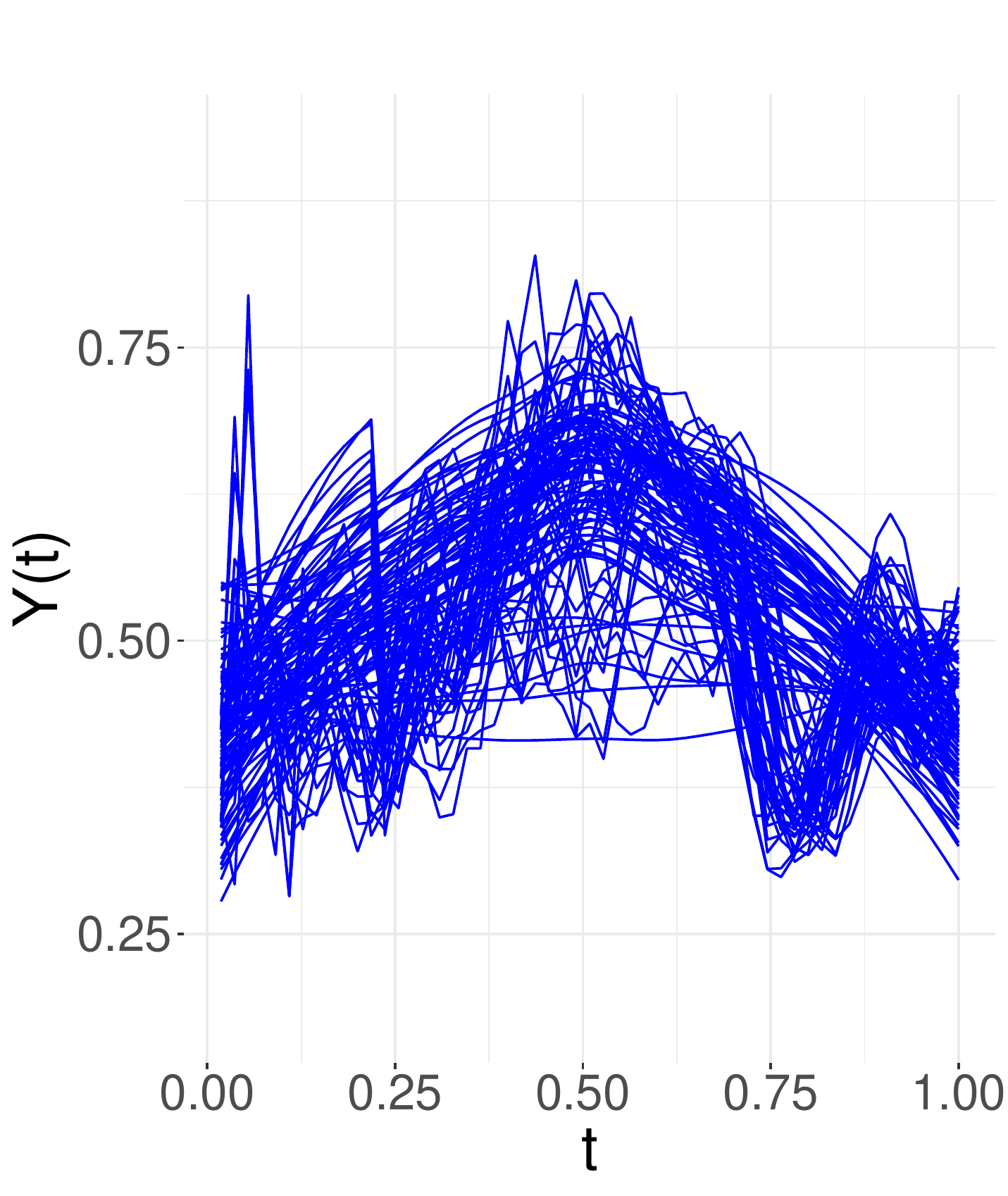}
		\caption{RCST FA tract profile ($Y$)}
		\label{fig:rcst}
	\end{subfigure}
	\caption{The first 100 pairs of CCA and RCST FA tract profiles.}
	\label{fig:dti}
\end{figure}

\begin{figure}[!t]
	\centering
	\begin{subfigure}{.49\textwidth}
		\centering
		\includegraphics[width=.7\textwidth, height=.25\textheight]
			{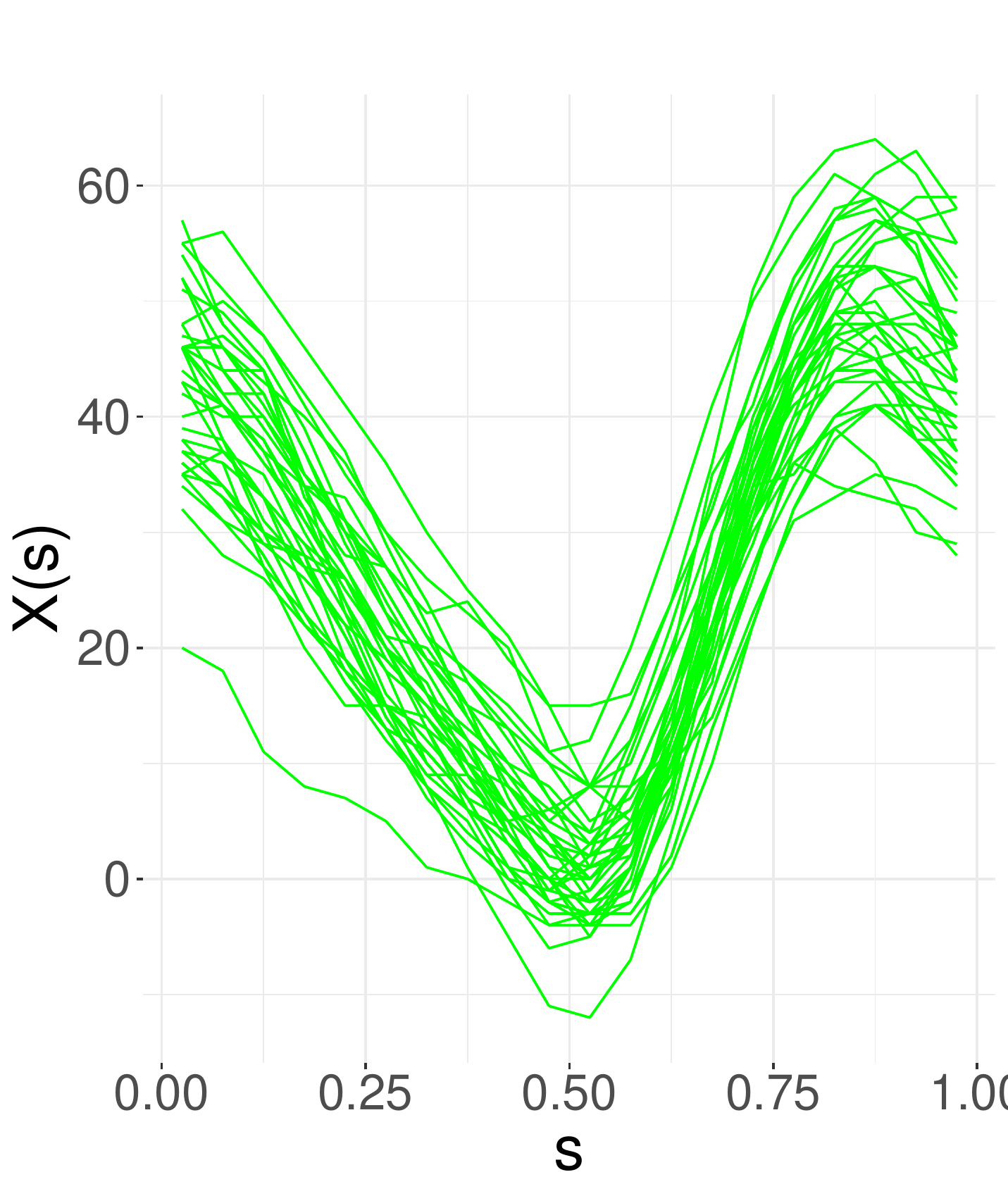}
		\caption{Hip angle ($X$)}
		\label{fig:hip}
	\end{subfigure}
	\begin{subfigure}{.49\textwidth}
		\centering
		\includegraphics[width=.7\textwidth, height=.25\textheight]
			{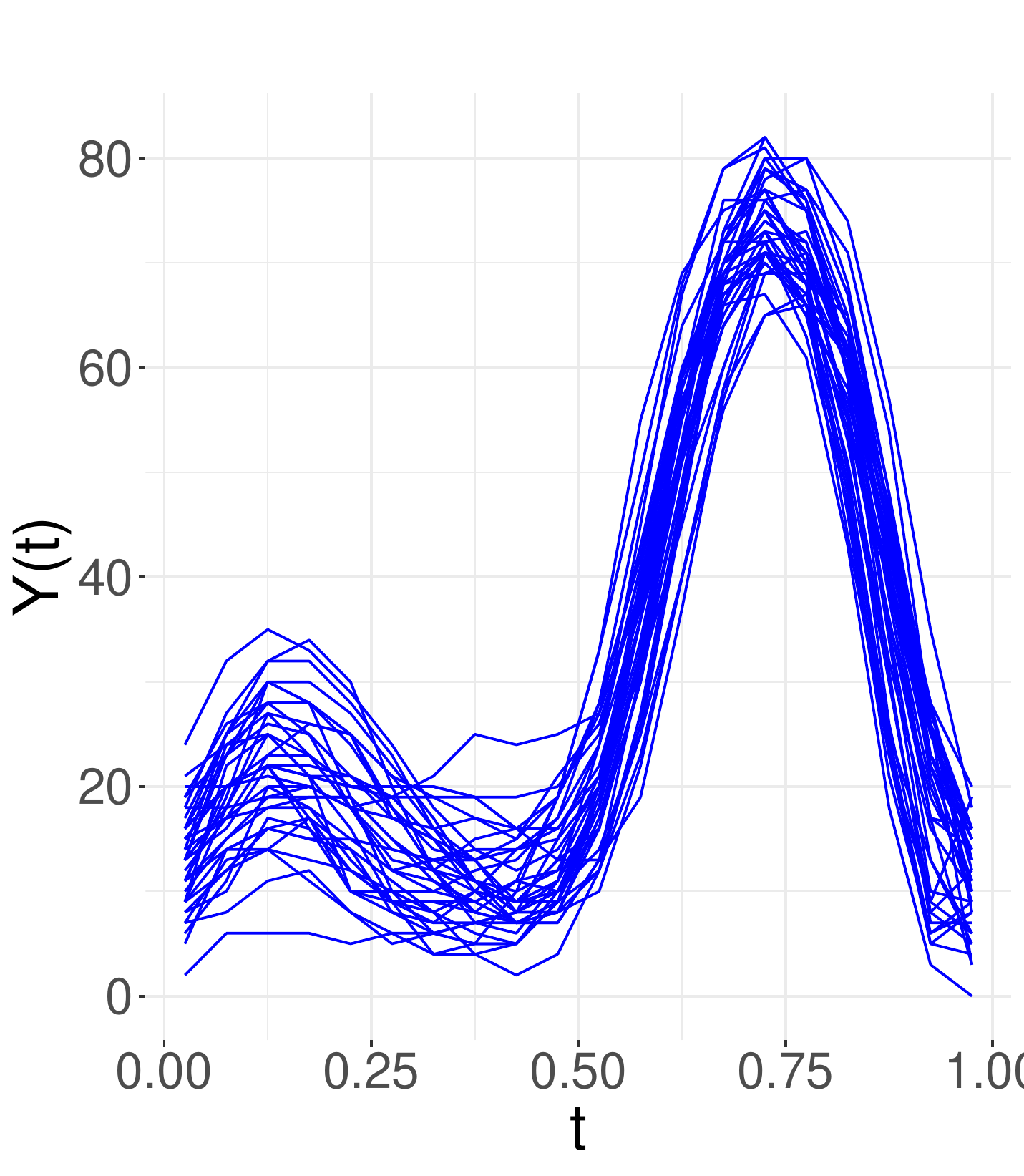}
		\caption{Knee angle ($Y$)}
		\label{fig:knee}
	\end{subfigure}
	\caption{Curves on boys' hip and knee angles.}
	\label{fig:bg}
\end{figure}

Analogous to simulation studies,
the real data analysis ended up with 
slight difference among ReISPEs averages (resp. standard deviations),
implying again the close accuracy of four competitors in prediction.
fAPLS consumed much less time in analyzing DTI data,
while this advantage did not exist for BG data.
We guessed the small sample size ($n=39$) of BG data 
saved the computational burden of SigComp.

\begin{table}[!t]
	\centering\scriptsize
	\caption{\small
		The averages$\times 100$ (and standard deviations$\times 100$) 
		of ReISPEs and ReISEEs in numerical experiments
		(running on a laptop with 
		AMD\textsuperscript{\textregistered} Ryzen\texttrademark\ 5 4500U @$6\times 2.38$ GHz with 16 GB RAM).
		Values of $\rho$ and $\sigma_\varepsilon^2$ were designed,
		whereas SNR was computed accordingly.
		Row minimums are underlined.
	}
	\label{tab:error}
	\begin{tabular}{cccc|rrrr}
	\hline
	& $\rho$ & $\sigma_\varepsilon^2$ & SNR 
	& \multicolumn{1}{c}{fAPLS} 
	& \multicolumn{1}{c}{SigComp} 
	& \multicolumn{1}{c}{NIPALS} 
	& \multicolumn{1}{c}{SIMPLS} \\ \hline
	\multicolumn{8}{c}{Estimation error: mean ReISEE$\times 100$ (standard deviation$\times 100$)} \\
	Simulation 1 & 0.1 & 1 & 10 & \underline{33.04} (14.53) & 72.63 (25.77) & 73.02  (7.12) & 42.18 (18.53) \\
	 &  & 100 & 1 & \underline{37.84}  (8.80) & 84.12 (21.39) & 73.71  (5.96) & 45.41 (14.27) \\
	 & 0.9 & 1 & 11 & \underline{33.20} (14.64) & 72.60 (27.28) & 71.82  (8.49) & 42.21 (18.57) \\
	 &  & 100 & 1 & \underline{37.84} (10.22) & 86.67 (23.26) & 73.78  (6.53) & 44.44 (16.77) \\
 	Simulation 2 & 0.1 & 1 & 7 & \underline{0.89} (0.62) & 1.35 (0.48) & 7.44 (1.33) & 0.95 (0.49) \\
	 &  & 80 & 1 & 13.01 (3.23) & 13.05 (5.50) & 21.00 (6.69) & \underline{12.75} (3.79) \\
	 & 0.9 & 1 & 7 & \underline{1.42} (2.02) & 1.68 (0.75) & 7.90 (1.63) & 1.64 (1.72) \\
	 &  & 80 & 1 & \underline{17.46} (15.45) & 21.03 (16.43) & 26.55 (12.73) & 23.46 (20.67) \\
 	Simulation 3 & 0.1 & 0.05 & 7 & 2.33 (5.92) & 8.98 (21.21) & 4.28 (8.54) & \underline{1.99} (5.75) \\
	 &  & 1 & 2 & \underline{7.04} (4.25) & 19.32 (26.67) & 13.66 (20.08) & 8.70 (15.92)  \\
	 & 0.9 & 0.05 & 7 & 6.27 (14.55) & 6.62 (18.34) & 6.22 (11.46) & \underline{5.33} (10.82) \\
	 &  & 1 & 2 & 23.53 (30.15) & 20.36 (27.07) & 27.88 (32.14) & \underline{16.68} (26.25) \\
	\multicolumn{8}{c}{Prediction error: mean ReISPE$\times 100$ (standard deviation$\times 100$)} \\
	Simulation 1 & 0.1 & 1 & 10 & \underline{1.65} (0.39) & 1.74 (0.43) & 1.94 (0.43) & 1.80 (0.45) \\
	 &  & 100 & 1 & 47.69 (5.37) & 47.69 (5.40) & 47.70 (5.37) & \underline{47.62} (5.37) \\
	 & 0.9 & 1 & 11 & \underline{1.56} (0.36) & 1.62 (0.43) & 1.85 (0.48) & 1.71 (0.42) \\
	 &  & 100 & 1 & \underline{46.57} (6.10) & \underline{46.57} (6.09) & 46.71 (6.08) & 46.61 (6.07) \\
 	Simulation 2 & 0.1 & 1 & 7 & 2.69 (0.50) & \underline{2.57} (0.49) & 2.69 (0.50) & 2.69 (0.49) \\
	 &  & 80 & 1 & 68.88 (5.33) & \underline{68.80} (5.64) & 68.89 (5.33) & 68.90 (5.40) \\
	 & 0.9 & 1 & 7 & 2.68 (0.63) & \underline{2.56} (0.62) & 2.70 (0.63) & 2.70 (0.63) \\
	 &  & 80 & 1 & 68.90 (6.35) & 69.06 (6.48) & \underline{68.88} (6.39) & 69.14 (6.43) \\
 	Simulation 3 & 0.1 & 0.05 & 7 & 28.68 (4.32) & 28.62 (4.43) & 28.61 (4.39) & \underline{28.56} (4.42) \\
	 &  & 1 & 2 & 90.18 (3.48) & 90.11 (3.38) & 90.21 (3.59) & \underline{90.04} (3.50) \\
	 & 0.9 & 0.05 & 7 & 28.37 (5.41) & \underline{28.11} (5.57) & 28.33 (5.36) & 28.37 (5.38) \\
	 &  & 1 & 2 & 89.36 (4.40) & \underline{89.16} (4.92) & 89.36 (4.48) & 89.22 (4.70) \\
 	DTI & - & - & - & 81.10 (3.51) & 81.97 (3.94) & 80.99 (3.69) & \underline{80.75} (3.59) \\
	BG & - & - & - & \underline{62.65} (12.68) & 71.23 (15.64) & 69.41 (16.12) & 65.04 (15.83) \\
	 \hline
	\end{tabular}
\end{table}

\begin{table}[!t]
	\centering\scriptsize
	\caption{\small
		Averages (and standard deviations) of component numbers 
		and total running time (in seconds)
		for numerical experiments
		(running on a laptop with 
		AMD\textsuperscript{\textregistered} Ryzen\texttrademark\ 5 4500U @$6\times 2.38$ GHz with 16 GB RAM).
		Values of $\rho$ and $\sigma_\varepsilon^2$ were designed,
		whereas SNR was computed accordingly.
		Row minimums are underlined.
	}
	\label{tab:ancillary}
	\begin{tabular}{cccc|rrrr}
	\hline
	& $\rho$ & $\sigma_\varepsilon^2$ & SNR 
	& \multicolumn{1}{c}{fAPLS} 
	& \multicolumn{1}{c}{SigComp} 
	& \multicolumn{1}{c}{NIPALS} 
	& \multicolumn{1}{c}{SIMPLS} \\ \hline
	\multicolumn{8}{c}{Number of components: average number (standard deviation)} \\
	Simulation 1 & 0.1 & 1 & 10 & 2.2  (0.4) & 2.2  (0.4) & \underline{2.1} (0.3) & 2.2 (0.4) \\
	 &  & 100 & 1 & \underline{2.1} (0.3) & \underline{2.1} (0.4) & \underline{2.1} (0.3) & \underline{2.1} (0.4) \\
	 & 0.9 & 1 & 11 & \underline{2.2} (0.4) & \underline{2.2} (0.4) & \underline{2.2} (0.4) & \underline{2.2} (0.4) \\
	 &  & 100 & 1 & \underline{2.1} (0.3) & \underline{2.1} (0.4) & \underline{2.1} (0.3) & 2.2 (0.4) \\
 	Simulation 2 & 0.1 & 1 & 7 & 6.2 (1.0) & \underline{3.1} (0.2) & 10.3  (1.1) & 9.9  (1.2) \\
	 &  & 80 & 1 & \underline{2.0} (0.3) & 3.5 (1.1) & 4.3 (1.3) & 4.3 (1.3) \\
	 & 0.9 & 1 & 7 & 6.5  (1.8) & \underline{3.0} (0.0) & 10.2 (1.4) & 10.0 (1.7) \\
	 &  & 80 & 1 & \underline{2.6} (1.9) & 4.7 (1.3) & 4.6 (1.8) & 5.2 (2.7) \\
 	Simulation 3 & 0.1 & 0.05 & 7 & 1.9 (0.5) & \underline{1.4} (0.6) & 2.2 (0.6) & 1.8 (0.7) \\
	 &  & 1 & 2 & \underline{1.3} (0.5) & \underline{1.3} (0.5) & 2.1 (0.4) & 1.4 (0.6) \\
	 & 0.9 & 0.05 & 7 & 2.1 (0.7) & \underline{1.5} (0.8) & 2.3 (0.5) & 2.0 (0.9) \\
	 &  & 1 & 2 & 1.7 (0.9) & \underline{1.4} (0.7) & 2.3 (0.6) & 1.5 (0.8) \\
 	DTI & - & - & - & 4.1 (0.8) & \underline{3.3} (0.7) & 4.6 (0.5) & 4.8 (0.4) \\
	BG & - & - & - & \underline{2.7} (0.9) & 4.1 (1.4) & 5.0 (1.6) & 4.9 (1.5) \\
	\multicolumn{8}{c}{Total running time in seconds for all replicates/splits} \\
	Simulation 1 & 0.1 & 1 & 10 & \underline{3.0} & 41.8 & 350.3 & 150.2  \\
	 &  & 100 & 1 & \underline{3.4} & 41.4 & 358.5 & 149.6 \\
	 & 0.9 & 1 & 11 & \underline{3.4} & 41.2 & 327.2 & 148.7 \\
	 &  & 100 & 1 & \underline{3.4} & 40.1 & 327.3 & 147.5 \\
 	Simulation 2 & 0.1 & 1 & 7 & \underline{36.1} & 48.6 & 398.5 & 241.7  \\
	 &  & 80 & 1 & \underline{35.8} & 51.4 & 416.3 & 242.2 \\
	 & 0.9 & 1 & 7 & \underline{35.7} & 50.1 & 375.0 & 242.3 \\
	 &  & 80 & 1 & \underline{36.4} & 49.3 & 377.7 & 242.9 \\
 	Simulation 3 & 0.1 & 0.05 & 7 & \underline{6.9} & 41.6 & 327.9 & 163.2 \\
	 &  & 1 & 2 & \underline{6.4} & 44.0 & 336.8 & 164.4 \\
	 & 0.9 & 0.05 & 7 & \underline{6.4} & 41.0 & 272.0 & 162.8 \\
	 &  & 1 & 2 & \underline{6.4} & 41.6 & 275.7 & 163.6 \\
 	DTI & - & - & - & \underline{5.4} & 71.8 & 266.4 & 101.6  \\
	BG & - & - & - & \underline{1.6} & 1.7 & 29.2 & 23.5  \\\hline
	\end{tabular}
\end{table}

\section{Conclusion \& discussion}\label{sec:discussion}

Fitting FoFR,
we suggest fAPLS,
a route of FPLS via Krylov subspaces.
The fAPLS estimator owns a concise and explicit expression.
Meanwhile,
we introduce an alternative but equivalent version,
stabilizing numerical outputs.
In spite of its less computational burden,
fAPLS is fairly competitive to existing FPLS routes in terms of both estimation and prediction accuracy.
Our proposal is potential to be further extended to more complex data structure,
as illustrated in the following paragraphs.

More efforts can be put on the estimation of $r_{XX}$ and $r_{XY}$.
To accommodate measurement errors,
the local linear smoothing
\citep[see, e.g.,][]{YaoMullerWang2005a, LiHsing2010}
and spline smoothing 
\citep[see, e.g.,][]{XiaoLiCheckleyCrainiceanu2018}
may be helpful.
In the case of geographic data, 
the spatial correlation 
(i.e., $X_i$ and $X_{i'}$, $i\neq i'$, no longer mutually independent)
lead to a potential inconsistency of PLS estimators; 
see \citet[Theorem~1]{SingerKrivobokovaMunkDeGroot2016} for this issue in the multivariate context.
A naive correction, 
transplanted from \citet[Section~4.1]{SingerKrivobokovaMunkDeGroot2016},
is to instead implement the regression on transformed observations $(X_i^*, Y_i^*)$, 
$i=1,\ldots,n$, 
such that, for all $(s,t)\in\mathbb I_X\times\mathbb I_Y$,
$[X_1^*(s),\ldots,X_n^*(s)]^\top=\bm V_{XX}^{-1/2}(s)[X_1(s),\ldots,X_n(s)]^\top$
and 
$[Y_1^*(t),\ldots,Y_n^*(t)]^\top=\bm V_{YY}^{-1/2}(t)[Y_1(t),\ldots,Y_n(t)]^\top$,
with matrices
$\bm V_{XX}(s)=[\cov\{X_i(s),X_{i'}(s)\}]_{n\times n}$
and $\bm V_{YY}(t)=[\cov\{Y_i(t),Y_{i'}(t)\}]_{n\times n}$.
But it is even challenging to recover $\bm V_{XX}$ and $\bm V_{YY}$ sufficiently accurately
without specifying the dependence structure,
since there is only one observation for each $i$.
Alternatively and more practically,
one can target at correcting naive $\hat r_{XX}$ and $\hat r_{XY}$ for dependent subjects;
\citet{PaulPeng2011} offered a solution to it.

fAPLS has got a heuristic extension to
multiple functional covariates, 
i.e., associated with each realization $Y_i\sim Y$,
there are $m>1$ functional covariates,
say $X_{ij}\sim X_{\cdot j}$, $1\leq j\leq m$,
and correspondingly $m$ coefficient functions $\beta^{*(j)}$, $1\leq j\leq m$.
In particular,
$$
    Y_i(t)=\mu_Y(t)+
        \sum_{i=1}^m\mathcal L_{X_{ij}}(\beta^{*(j)}) + \varepsilon_i(t),
$$
where $Y_i$ and $X_{ij}$ are assumed to be independent across all $i$.
Following the idea of \eqref{eq:beta.p.fapls},
an ad hoc estimator for $(\beta^{*(1)},\ldots,\beta^{*(m)})$ is thus
\begin{multline*}
    (\hat\beta_{\rm fAPLS}^{(1)},\ldots,\hat\beta_{\rm fAPLS}^{(m)})
    =\argmin_{
        \beta^{(j)}\in {\rm KS}_p(\widehat\Gamma_{X_{\cdot j}X_{\cdot j}}, \beta^{*(j)}),
        \ 1\leq j\leq m
    }\frac{1}{m}\sum_{i=1}^m\int_{\mathbb I_Y}
        \Bigg\{Y_i(t)-\bar Y_i(t)
    \\
            -\sum_{j=1}^m\int_{\mathbb I_{X_{\cdot j}}}(X_{ij}-\bar X_{\cdot j})(s)\beta^{(j)}(s,t)\dd s
        \Bigg\}^2\dd t,
\end{multline*}
with $\bar X_{\cdot j}=m^{-1}\sum_{j=1}^m X_{ij}$
and domains $\mathbb I_{X_{\cdot j}}$ varying with $j$.
Of course,
it becomes necessary to introduce penalties once the above minimizer is not uniquely defined. 

Although fAPLS appears to merely work for linear models,
it is possible to be utilized in fitting the (functional) generalized linear models 
and proportional hazard models.
The basic idea, 
inherited from \cite{Marx1996},
is to embed PLS routes into iteratively reweighted least squares \citep{Green1984} in maximizing likelihood.
Recent successful applications of this idea include \cite{Albaqshi2017} and \cite{WangIbrahimZhu2020}.
Once trajectories are sparsely and irregularly observed,
fAPLS may be further modified in analogy to \cite{ZhouLockhart2020}.

\section*{Acknowledgment}

Special thanks go to Professor Richard A. Lockhart at Simon Fraser University for his constructive suggestions.
The author's work was financially supported by the Natural Sciences and Engineering Research Council of Canada (NSERC).

\bigskip
\appendix

\section{Appendix}\label{appendix:tech}

\setcounter{Lemma}{0}
\renewcommand{\theLemma}{\Alph{section}.\arabic{Lemma}}

In detail our assumptions are summarized as below.
\begin{enumerate}[label=(C\arabic*), resume]
	\setcounter{enumi}{\value{enumi}}
	\item\label{cond:identifiable}
		$\sum_{j, j'=1}^\infty\lambda_{j,X}^{-2}\left\{
				\int_{\mathbb I_Y}\int_{\mathbb I_X}\phi_{j,X}(s)r_{XY}(s,t)\phi_{j',Y}(t)\dd s\dd t
			\right\}^2<\infty$.
		$\beta^*$ belongs to 
		${\rm range}(\Gamma_{XX})
		=\{\Gamma_{XX}(f)\mid f\in L_2(\mathbb{I}_X\times\mathbb{I}_Y)\}$.
	\item\label{cond:moment}
		$\E(\|X\|_2^4)<\infty$ for all $t\in\mathbb I_Y$.
	\item\label{cond:bound.p}
	    As $n\to\infty$, $p=p(n)=O(n^{1/2})$.
	\item\label{cond:inf.norm}
	    Let $\mathbb I_X=[0,1]$.
		Both $\|\xi_{XX}\|_{\infty,2}$ and $\|\eta_{XX}\|_{\infty,2}$ are of order $O_p(1)$ as $n\to\infty$,
		with $\xi_{XX}$ and $\eta_{XX}$ defined as in \autoref{lemma:conv.gamma.1}
		and $\|\cdot\|_{\infty,2}$ defined such that 
		$\|f\|_{\infty,2}=\sup_{s\in \mathbb I_X}\{\int_{\mathbb I_X}f^2(s,t)\dd t\}^{1/2}$ 
		for $f\in L_2(\mathbb I_X\times\mathbb I_X)$.
	\item\label{cond:diverge.p.add.L2}
        Additional requirements on $p$ vary with the magnitude of $\|r_{XX}\|_2$; they also depend on $\tau_p$,
        the smallest eigenvalue of $\bm H_p$.
        \begin{itemize}
            \item If $\|r_{XX}\|_2\geq 1$,
		        then, as $n\to\infty$, 
		        $n^{-1}\tau_p^{-2}p^4\|r_{XX}\|_2^{4p}\max(1, \tau_p^{-2}p^2\|r_{XX}\|_2^{4p})$
		        and $n^{-1}\tau_p^{-3}p^{5}\|r_{XX}\|_2^{6p}$
		        are both of order $o(1)$;
			\item if $\|r_{XX}\|_2<1$,
		        then
		        $(n\tau_p^4)^{-1}=o(1)$ as $n$ diverges.
		\end{itemize}
    \item\label{cond:diverge.p.add.sup}
	    Keep everything in \ref{cond:diverge.p.add.L2}
	    but substitute $\|r_{XX}\|_\infty$ for $\|r_{XX}\|_2$.
	    Meanwhile,
	    require that $\|\beta_{p,\rm fAPLS}-\beta^*\|_{\infty}=o(1)$ as $p$ diverges,
	    viz. an enhanced version of \autoref{prop:conv.beta.p}.
	\item\label{cond:moment.add.sup}
	    Stochastic process $Y$ is ``eventually totally bounded in mean''
	    \citep[as defined by][(5)--(7)]{Hoffmann-Jorgensen1985};
	    i.e.,
	    in our context,
        \begin{itemize}
            \item $\E(\|Y\|_{\infty})<\infty$;
			\item for each $\epsilon>0$,
        	    there is a finite cover of $\mathbb{T}$,
        	    say ${\rm Cover}(\mathbb{T})$, 
        	    for each set $\mathbb{A}\in{\rm Cover}(\mathbb{T})$,
        	    such that
        		$\inf_{n\in\mathbb{Z}^+}
        		    n^{-1}\E\{\sup_{t,t'\in\mathbb{A}}|Y(t)-Y(t')|\}<\epsilon$.
		\end{itemize}
\end{enumerate}

Introducing \ref{cond:identifiable},
\citet[Theorem~2.3]{HeMullerWangYang2010} confirmed the identifiability of $\beta^*$ for FoFR 
and derived \eqref{eq:beta*}.
\ref{cond:identifiable} was also the foundation of \cite{YaoMullerWang2005b}.
Assumptions \ref{cond:moment}--\ref{cond:inf.norm} 
are prerequisites for the convergence of 
$\widehat\Gamma_{XX}^j(\beta^*)$ ($=\widehat\Gamma_{XX}^{j-1}(\hat r_{XY})$) 
which is uniform in $j\geq 1$.
One may feel unclear about the technical conditions stated in \ref{cond:diverge.p.add.L2} for the scenario of $\|r_{XX}\|_2\geq 1$:
virtually a special case for is that
$n^{-1}\max(\tau_p^{-4}, \tau_p^{-6}, \tau_p^{-8})=o(1)$
and $p=O(\ln\ln n)$.
Apparently,
$p$ is more restricted when $\|r_{XX}\|_2\geq 1$ than 
in the case of $\|r_{XX}\|_2<1$ 
(for the latter case $p$ is allowed to diverge at the rate of $O(n^{1/2})$);
that is why \citet{DelaigleHall2012b} 
suggested changing the scale on which $X$ is measured.
\ref{cond:diverge.p.add.sup} is stronger than \ref{cond:diverge.p.add.L2},
enabling us to consider the $L_\infty$-convergence.
At last,
we add \ref{cond:moment.add.sup}
as a prerequisite for the uniform law of large numbers for $\{Y_i\mid i\geq 1\}$.

\begin{Lemma}\label{lemma:conv.gamma.1}
	For each $(s,s',t)\in \mathbb I_X\times\mathbb I_X\times\mathbb I_Y$,
	\begin{align}
		\hat r_{XX}(s,s') &= r_{XX}(s,s') + n^{-1/2}\xi_{XX}(s,s') + n^{-1}\eta_{XX}(s,s'),\notag\\
		\hat r_{XY}(s,t) &= r_{XY}(s,t) + n^{-1/2}\xi_{XY}(s,t) + n^{-1}\eta_{XY}(s,t)
	\end{align}
	where, with identity operator $I:\mathbb R\to\mathbb R$,
	\begin{align*}
		\xi_{XX}(s,s') &= \frac{1}{\sqrt n}\sum_{i=1}^{n}(I-\E)[\{X_i(s)-\mu_X(s)\}\{X_i(s')-\mu_X(s')\}],\\
		\eta_{XX}(s,s') &= -n\{\bar X(s)-\mu_X(s)\}\{\bar X(s')-\mu_X(s')\},\\
		\xi_{XY}(s,t) &= \frac{1}{\sqrt n}\sum_{i=1}^{n}(I-\E)[\{X_i(s)-\mu_X(s)\}\{Y_i(t)-\mu_Y(t)\}],\\
		\eta_{XY}(s,t) &= -n\{\bar X(s)-\mu_X(s)\}\{\bar Y(t)-\mu_Y(t)\},
	\end{align*}
	and $\|\xi_{XX}\|_2$, $\|\eta_{XX}\|_2$, $\|\xi_{XY}\|_2$ and $\|\eta_{XY}\|_2$ all equal $O_p(1)$ as $n$ diverges.
\end{Lemma}

\begin{proof}[Proof of \autoref{lemma:conv.gamma.1}]
	It is an immediate implication of \citet[(5.1)]{DelaigleHall2012b}.
\end{proof}

\begin{Lemma}\label{lemma:conv.gamma.i}
    Assume \ref{cond:identifiable} and \ref{cond:moment}
    and that there is $C>0$ such that, 
    for all $n$, 
    we have $p\leq Cn^{-1/2}$ (i.e., condition \ref{cond:bound.p}).
    Then, for each $\epsilon>0$, 
    there are positive $C_1$, $C_2$ and $n_0$ such that,
    for each $n> n_0$,
    $$
        \Pr\left[
            \bigcap_{j=1}^{p}\left\{
                \|\widehat\Gamma_{XX}^j(\beta^*)-\Gamma_{XX}^j(\beta^*)\|_2
                \leq n^{-1/2}\|r_{XX}\|_2^{j-1}\{C_1 + C_2(j-1)\}
            \right\}
        \right]\geq 1-\epsilon.
	$$
	Assuming one more condition \ref{cond:inf.norm},
    $$
        \Pr\left[\bigcap_{j=1}^{p}\left\{
            \|\widehat\Gamma_{XX}^j(\beta^*)-\Gamma_{XX}^j(\beta^*)\|_{\infty}
            \leq n^{-1/2}\|r_{XX}\|_{\infty}^{i-1}\{C_1+C_2(j-1)\}
        \right\}\right]
        \geq 1-\epsilon.
    $$
\end{Lemma}

\begin{proof}[Proof of \autoref{lemma:conv.gamma.i}]
    Since $\Gamma_{XX}(\beta^*)=r_{XY}$ and $\widehat\Gamma_{XX}(\beta^*)=\hat r_{XY}$,
    \autoref{lemma:conv.gamma.i} is simply implied by \autoref{lemma:conv.gamma.1}
    when $p=1$.
    For integer $j\geq 2$ and each $(s,s',t)\in\mathbb I_X\times\mathbb I_X\times\mathbb I_Y$,
	\begin{align*}
	  |\widehat\Gamma_{XX}^j(\beta^*)(s,t)&-\Gamma_{XX}^j(\beta^*)(s,t)|
	  \\
	  =&\ |
	      \widehat\Gamma_{XX}\{\widehat\Gamma_{XX}^{j-1}(\beta^*)-\Gamma_{XX}^{j-1}(\beta^*)\}(s,t)
	      +(\widehat\Gamma_{XX}-\Gamma_{XX})\{\Gamma_{XX}^{j-1}(\beta^*)\}(s,t)
	  |
	  \\
	  \leq&\ \left\{\int_{\mathbb I_X}\hat r_{XX}^2(s,s')\dd w\right\}^{1/2}
	  \left[
	  	\int_{\mathbb I_X}\{\widehat\Gamma_{XX}^{j-1}(\beta^*)-\Gamma_{XX}^{j-1}(\beta^*)\}(s',t)\dd w
	  \right]^{1/2}
	  \\
	  &+ \left[\int_{\mathbb I_X}\{\hat r_{XX}(s,s')-r_{XX}(s,s')\}^2\dd s'\right]^{1/2}
	     	\left\{\int_{\mathbb I_X}\Gamma_{XX}^{j-1}(\beta^*)(s',t)\dd s'\right\}^{1/2}.
	\end{align*}
    It implies that,
    by the triangle inequality,
    $$
        \|\widehat\Gamma_{XX}^j(\beta^*)-\Gamma_{XX}^j(\beta^*)\|_2
        \leq
        \|\hat r_{XX}\|_2
            \|\widehat\Gamma_{XX}^{j-1}(\beta^*)-\Gamma_{XX}^{j-1}(\beta^*)\|_2
        +
        \|\hat r_{XX}-r_{XX}\|_2\|\Gamma_{XX}^{j-1}(\beta^*)\|_2.
    $$
    On iteration it gives that
    \begin{multline}\label{eq:bound.gamma.j.hat.L2}
        \|\widehat\Gamma_{XX}^j(\beta^*)-\Gamma_{XX}^j(\beta^*)\|_2
        \leq
        \|\hat r_{XX}\|_2^{j-1}\|
            \widehat\Gamma_{XX}(\beta^*)-\Gamma_{XX}(\beta^*)\|_2
        \\
        +
        \|\hat r_{XX}-r_{XX}\|_2\sum_{j'=1}^{j-1}
            \|\hat r_{XX}\|_2^{j-j'-1}\|\Gamma_{XX}^{j'}(\beta^*)\|_2.
    \end{multline}
    For each $\epsilon>0$,
    there is $n_0>0$ such that,
    for all $n>n_0$,
    we have
    \begin{align*}
        1-\epsilon/2
        \leq&\ \Pr(\|\hat r_{XX}-r_{XX}\|_2\leq C_0n^{-1/2})
        \leq\Pr(\|\hat r_{XX}\|_2\leq \|r_{XX}\|_2+C_0n^{-1/2}),
        \\
        1-\epsilon/2
        \leq&\ \Pr(\|\hat r_{XY}-r_{XY}\|_2\leq C_0n^{-1/2}),
    \end{align*}
    with constant $C_0>0$, by \autoref{lemma:conv.gamma.1}.
    It follows \eqref{eq:bound.gamma.j.hat.L2} that
    \begin{align*}
        1-\epsilon
        \leq\Pr\Bigg[\bigcap_{j=1}^{p}\Bigg[
            \|(\widehat\Gamma_{XX}^j&-\Gamma_{XX}^j)(\beta^*)\|_2
            \leq C_0n^{-1/2}\Bigg\{
	            (\|r_{XX}\|_2+C_0n^{-1/2})^{j-1}
	    \\
	            &+
	            \sum_{j'=1}^{j-1}
	                \|r_{XX}\|_2^{j'}\|\beta^*\|_2(\|r_{XX}\|_2+C_0n^{-1/2})^{j-j'-1}
        \Bigg\}\Bigg]\Bigg]
        \\
        \leq\Pr\Bigg[\bigcap_{j=1}^{p}\Bigg[
            \|(\widehat\Gamma_{XX}^j&-\Gamma_{XX}^j)(\beta^*)\|_2  
            \leq C_0n^{-1/2}\|r_{XX}\|_2^{j-1}\Bigg\{
            (1+C_0n^{-1/2}/\|r_{XX}\|_2)^{j-1}
        \\
            &+
            \|\beta^*\|_2\sum_{j'=1}^{j-1}
                (1+C_0n^{-1/2}/\|r_{XX}\|_2)^{j-j'-1}
        \Bigg\}\Bigg]\Bigg]
        \\
        \leq\Pr\Bigg[\bigcap_{j=1}^{p}\Bigg\{
            \|\widehat\Gamma_{XX}^j&(\beta^*)-\Gamma_{XX}^j(\beta^*)\|_2
        \\
            &\leq n^{-1/2}\|r_{XX}\|_2^{j-1}\{C_1+C_2(j-1)\}
        \Bigg\}\Bigg],
        \quad\text{(since $p\leq Cn^{1/2}$)}
    \end{align*}
    where $C_1=C_0\exp(CC_0/\|r_{XX}\|_2)$ 
    and $C_2=\|\beta^*\|_2C_1$.
    
    Suppose \ref{cond:inf.norm} holds.
    Similar to \eqref{eq:bound.gamma.j.hat.L2},
    \begin{align*}
        \|\widehat\Gamma_{XX}^j(\beta^*)-\Gamma_{XX}^j(\beta^*)\|_\infty
        \leq& \
        \|\hat r_{XX}\|_\infty^{j-1}\|
            \widehat\Gamma_{XX}(\beta^*)-\Gamma_{XX}(\beta^*)\|_\infty
        \\
        &+
        \|\hat r_{XX}-r_{XX}\|_\infty\sum_{j'=1}^{j-1}
            \|\hat r_{XX}\|_\infty^{j-j'-1}\|\Gamma_{XX}^{j'}(\beta^*)\|_\infty
        \\
        \leq& \
        \|\hat r_{XX}\|_\infty^{j-1}\|
            \widehat\Gamma_{XX}(\beta^*)-\Gamma_{XX}(\beta^*)\|_\infty
        \\
        &+
        \|\hat r_{XX}-r_{XX}\|_\infty\sum_{j'=1}^{j-1}
            \|\hat r_{XX}\|_\infty^{j-j'-1}\|r_{XX}\|_\infty^{j'}\|\beta^*\|_\infty.
    \end{align*}
    Mimicking the argument above for the $L_2$ sense,
    one obtains that
    $$
        \Pr\left[\bigcap_{j=1}^{p}\left\{
            \|\widehat\Gamma_{XX}^j(\beta^*)-\Gamma_{XX}^j(\beta^*)\|_{\infty}
            \leq n^{-1/2}\|r_{XX}\|_\infty^{j-1}\{C_1+C_2(j-1)\}
        \right\}\right]
        \geq 1-\epsilon,
    $$
    with, at this time, $C_1=C_0\exp(CC_0/\|r_{XX}\|_\infty)$ and $C_2=\|\beta^*\|_\infty C_1$.
    The finiteness of $\|\beta^*\|_\infty$ originates from the continuity of eigenfunctions $\phi_{i,X}$'s and $\phi_{i,Y}$'s
    (refer to the Mercer's theorem).
\end{proof}

\begin{proof}[Proof of \autoref{prop:conv.beta.p}]
	Recall $\beta_{p,p',\rm FPCR}$ at \eqref{eq:beta.p.q.fpca}
	and introduce $\beta_{p,\infty,\rm FPCR}\in L_2(\mathbb I_X\times \mathbb I_Y)$ 
	such that
	$$
		\beta_{p,\infty,\rm FPCR}(s,t)
		=\lim_{p'\to\infty}\beta_{p,p',\rm FPCR}(s,t)
		= \sum_{j=1}^p\frac{\phi_{j,X}(s)}{\lambda_{j,X}}\int_{\mathbb I_X}\phi_{j,X}(s')r_{XY}(s',t)\dd s'.
	$$
	It follows that 
	$$
		\Gamma_{XX}(\beta_{p,\infty,\rm FPCR})(s,t)
		= \sum_{j=1}^p\phi_{j,X}(s)\int_{\mathbb I_X}\phi_{j,X}(s')r_{XY}(s',t)\dd w.
	$$
	Now 
	$$
		[
			(\lambda_{1,X}I-\Gamma_{XX})
			\circ\cdots\circ
			(\lambda_{p,X}I-\Gamma_{XX})
		]
		(\beta_{p,\infty,\rm FPCR})=0
	$$
	in which the left-hand side equals 
	$\sum_{i=j}^pa_j\Gamma_{XX}^j(\beta_{p,\infty,\rm FPCR})$ 
	with $a_0=\prod_{j=1}^p\lambda_{j,X}>0$.
	Therefore,
	$$
		\beta_{p,\infty,\rm FPCR} = -\sum_{j=1}^p\frac{a_j}{a_0}\Gamma_{XX}^j(\beta_{p,\infty,\rm FPCR}).
	$$
	Denote by $P_p:{\rm range}(\Gamma_{XX})\to {\rm range}(\Gamma_{XX})$ 
	the operator that projects elements in ${\rm range}(\Gamma_{XX})$ to 
	${\rm span}\{f_{jj'}\in L_2(\mathbb I_X\times \mathbb I_Y)\mid 
		f_{jj'}(s,t)=\phi_{j,X}(s)\phi_{j',Y}(t), 1\leq j\leq p, j'\geq 1\}$.
	Thus $\beta_{p,\infty,\rm FPCR}=P_p(\beta^*)$.
	Since
	$\Gamma_{XX}^j(\beta_{p,\infty,\rm FPCR})=P_p[\Gamma_{XX}^j(\beta^*)]$,
	one has
	$$
		P_p\left[\beta^*+\sum_{j=1}^p\frac{a_j}{a_0}\Gamma_{XX}^j(\beta^*)\right]=0,
	$$
	implying that, 
	for all $p$,
	$$
		P_p(\beta^*)
		\in\{P_p(f)\mid f\in\overline{{\rm KS}_\infty(\Gamma_{XX}, \beta^*)}\}.
	$$
	Taking limits as $p\rightarrow\infty$ on both sides of the above formula,
	we obtain $\beta^*\in\overline{{\rm KS}_\infty(\Gamma_{XX}, \beta^*)}$ and accomplish the proof.
\end{proof}

\begin{proof}[Proof of \autoref{prop:conv.beta.p.hat}]
	Recall $\beta_{p,\rm fAPLS}$ \eqref{eq:beta.p.fapls} and $\hat\beta_{p,\rm fAPLS}$ \eqref{eq:beta.p.fapls.hat}
	and notations in defining them.
	The Cauchy-Schwarz inequality implies that
	\begin{align*}
		|\hat h_{jj'}&-h_{jj'}|
		\\
		\leq&\ \|\widehat\Gamma_{XX}^j(\beta^*)-\Gamma_{XX}^j(\beta^*)\|_2
		\|\widehat\Gamma_{XX}^{j'+1}(\beta^*)\|_2
		+
		\|\widehat\Gamma_{XX}^{j'+1}(\beta^*)-\Gamma_{XX}^{j'+1}(\beta^*)\|_2
		\|\Gamma_{XX}^j(\beta^*)\|_2
		\\
		\leq&\ 
		\|\widehat\Gamma_{XX}^j(\beta^*)-\Gamma_{XX}^j(\beta^*)\|_2
		\|\hat r_{XX}\|_2^{j+1}\|\beta^*\|_2
		+
		\|\widehat\Gamma_{XX}^{j'+1}(\beta^*)-\Gamma_{XX}^{j'+1}(\beta^*)\|_2
		\|r_{XX}\|_2^j\|\beta^*\|_2.
	\end{align*}
	By Lemmas \ref{lemma:conv.gamma.1} and \ref{lemma:conv.gamma.i},
    for each $\epsilon>0$ and $p\leq Cn^{1/2}$,
    there are positive $n_0$, $C_3$ and $C_4$ such that,
    for all $n>n_0$,
	\begin{align*}
		1-\epsilon
		\leq\Pr\Bigg[\bigcap_{j,j'=1}^p\Big\{
			|\hat h_{jj'}-h_{jj'}|
			\leq& 
			\|\widehat\Gamma_{XX}^j(\beta^*)-\Gamma_{XX}^j(\beta^*)\|_2
			(\|r_{XX}\|_2+C_0n^{-1/2})^{j'+1}\|\beta^*\|_2
		\\
			&+
			\|\widehat\Gamma_{XX}^{j'+1}(\beta^*)-\Gamma_{XX}^{j'+1}(\beta^*)\|_2
			\|r_{XX}\|_2^j\|\beta^*\|_2
	  	\Big\}\Bigg]
		 \\
		 \leq\Pr\Bigg[\bigcap_{j,j'=1}^p\Big\{
		    |\hat h_{jj'}-h_{jj'}|
		    \leq&\ 
		    n^{-1/2}\|r_{XX}\|_2^{i+j'}\{C_3\max(j,j')+C_4\}
		 \Big\}\Bigg].
	\end{align*}
    Thus
    \begin{align}
        \|\widehat{\bm H}_p-\bm H_p\|_2^2
        \leq& \sum_{j,j'=1}^p|\hat h_{jj'}-h_{jj'}|^2
        \notag\\
        =&\ O_p\left(
                n^{-1}\sum_{j,j'=1}^p\|r_{XX}\|_2^{2j+2j'}
            \right)
        + O_p\left\{
                n^{-1}\sum_{j,j'=1}^p\max(j^2,j'^2)\|r_{XX}\|_2^{2j+2j'}
            \right\}
        \notag\\
        =&\ \begin{cases}
            O_p(n^{-1}p^2\|r_{XX}\|_2^{4p}) + O_p(n^{-1}p^4\|r_{XX}\|_2^{4p})
            &\text{if }\|r_{XX}\|_2\geq 1
            \\
            O_p(n^{-1})
            &\text{if }\|r_{XX}\|_2< 1
        \end{cases}
        \notag\\
        =&\ \begin{cases}
            O_p(n^{-1}p^4\|r_{XX}\|_2^{4p})
            &\text{if }\|r_{XX}\|_2\geq 1
            \\
            O_p(n^{-1})
            &\text{if }\|r_{XX}\|_2< 1.
        \end{cases}
        \label{eq:diff.H.p}
    \end{align}
	Here $\|\cdot\|_2$ is abused for the matrix norm induced by the Euclidean norm,
	i.e., for arbitrary $\bm A\in\mathbb R^{p\times p'}$ and $\bm b\in\mathbb R^{p'\times 1}$
	$\|\bm A\|_2=\sup_{\bm b:\|\bm b\|_2=1}\|\bm A\bm b\|_2$
	is actually the largest eigenvalue of $\bm A$.
	It reduces to the Euclidean norm for vectors.
    It is analogous to \eqref{eq:diff.H.p} to deduce that
    \begin{equation}\label{eq:diff.alpha.j}
        \|\widehat{\bm\alpha}_p-\bm\alpha_p\|_2^2
        = \sum_{j=1}^p|\hat{\alpha}_j-\alpha_j|^2
        =\ \begin{cases}
            O_p(n^{-1}p^3\|r_{XX}\|_2^{2p})
            &\text{if }\|r_{XX}\|_2\geq 1
            \\
            O_p(n^{-1})
            &\text{if }\|r_{XX}\|_2< 1.
        \end{cases}
    \end{equation}
    Denote by $\tau_p$ the smallest eigenvalue of $\bm H_p$.
    Noting that $\|\bm H_p^{-1}\|_2=\tau_p^{-1}$,
    for $p\leq Cn^{1/2}$,
    $$
        \|(\widehat{\bm H}_p-\bm H_p)\bm H_p^{-1}\|_2
        \leq \tau_p^{-1}\|\widehat{\bm H}_p-\bm H_p\|_2
        =\begin{cases}
            O_p(n^{-1/2}\tau_p^{-1}p^2\|r_{XX}\|_2^{2p})
            &\text{if }\|r_{XX}\|_2\geq 1
            \\
            O_p(n^{-1/2}\tau_p^{-1})
            &\text{if }\|r_{XX}\|_2< 1.
        \end{cases}
    $$
    Introduce random matrix $\bm M_p\in\mathbb R^{p\times p}$ such that
    $\bm I-\bm H_p^{-1}(\widehat{\bm H}_p-\bm H_p)+\bm M_p=\{\bm I+\bm H_p^{-1}(\widehat{\bm H}_p-\bm H_p)\}^{-1}$,
    i.e.,
    $
    	\bm M_p
    	=\{\bm I+\bm H_p^{-1}(\widehat{\bm H}_p-\bm H_p)\}^{-1}\bm H_p^{-1}(\widehat{\bm H}_p-\bm H_p)\bm H_p^{-1}(\widehat{\bm H}_p-\bm H_p).
    $
    Therefore,
    $$
    	\|\bm M_p\|_2
    	\leq \|\bm I+\bm H^{-1}(\widehat{\bm H}_p-\bm H_p)\|_2^{-1}\|\bm H^{-1}(\widehat{\bm H}_p-\bm H_p)\|_2^2
    	\leq (1-\rho)^{-1}\tau_p^{-2}\|\widehat{\bm H}_p-\bm H_p\|_2^2,
    $$
    provided that $\tau_p^{-1}\|\widehat{\bm H}_p-\bm H_p\|_2\leq\rho<1$ 
    (refer to \citealp[][(7.18)]{DelaigleHall2012b}).
    Revealed by the identity that
    $\widehat{\bm H}_p^{-1}=\{\bm I+\bm H_p^{-1}(\widehat{\bm H}_p-\bm H_p)\}^{-1}\bm H_p^{-1}$,
    \begin{align}
    	\|\widehat{\bm H}_p^{-1}&-\bm H_p^{-1}\|_2
    	\notag\\
    	\leq&\ \{\|\bm H_p^{-1}(\widehat{\bm H}_p-\bm H_p)\|_2+\|\bm M_p\|_2\}\|\bm H_p^{-1}\|_2
    	\notag\\
    	=&\ \begin{cases}
    		O_p(n^{-1/2}\tau_p^{-2}p^2\|r_{XX}\|_2^{2p})
    		+
    		O_p(n^{-1}\tau_p^{-3}p^4\|r_{XX}\|_2^{4p})
		 	&\text{if }\|r_{XX}\|_2\geq 1
		 	\\
		 	O_p(n^{-1/2}\tau_p^{-2})+O_p(n^{-1}\tau_p^{-3})
		 	&\text{if }\|r_{XX}\|_2< 1.
		\end{cases}
		\label{eq:diff.H.p.inv}
    \end{align}
    Combining \eqref{eq:diff.alpha.j}, \eqref{eq:diff.H.p.inv} and the identity that
    \begin{align}
        \|\bm\alpha_p\|_2
        =&\ \left[\sum_{j=1}^p\left\{\int_{\mathbb I_Y}\int_{\mathbb I_X}r_{XY}(s,t)\Gamma_{XX}^j(\beta^*)(s,s')\dd s\dd s'\right\}^2\right]^{1/2}
        \notag\\
        \leq&\ \left[\sum_{j=1}^p\|r_{XY}\|_2^2\|\Gamma_{XX}^j(\beta^*)\|_2^2\right]^{1/2}
        \notag\\
        =&\ \begin{cases}
            O(p^{1/2}\|r_{XX}\|_2^p)
            &\text{if }\|r_{XX}\|_2\geq 1
            \\
            O(1)
            &\text{if }\|r_{XX}\|_2< 1,
        \end{cases}
        \label{eq:bound.alpha.p}
    \end{align}
    we reach that
    \begin{align}
        \|\widehat{\bm H}_p^{-1}&\widehat{\bm\alpha}_p-\bm H_p^{-1}\bm\alpha_p\|_2
        \notag\\
        \leq&\ \|\widehat{\bm H}_p^{-1}\|_2\|\widehat{\bm\alpha}_p-\bm\alpha_p\|_2
            +\|\widehat{\bm H}_p^{-1}-\bm H_p^{-1}\|_2\|\bm\alpha_p\|_2
        \notag\\
        =&\ \begin{cases}
        	O_p(n^{-1/2}\tau_p^{-1}p^{3/2}\|r_{XX}\|_2^p)
        	\notag\\
            \qquad +\ O_p(n^{-1/2}\tau_p^{-2}p^{5/2}\|r_{XX}\|_2^{3p})
            + O_p(n^{-1}\tau_p^{-3}p^{9/2}\|r_{XX}\|_2^{5p})
            &\text{if }\|r_{XX}\|_2\geq 1
            \\
            O_p(n^{-1/2}\tau_p^{-1})
            +O_p(n^{-1/2}\tau_p^{-2})
            +O_p(n^{-1}\tau_p^{-3})
            &\text{if }\|r_{XX}\|_2< 1
        \end{cases}
        \notag\\
        =&\ \begin{cases}
        	O_p(n^{-1/2}\tau_p^{-1}p^{3/2}\|r_{XX}\|_2^p)
        	\\
            \qquad +\ O_p(n^{-1/2}\tau_p^{-2}p^{5/2}\|r_{XX}\|_2^{3p})
            + O_p(n^{-1}\tau_p^{-3}p^{9/2}\|r_{XX}\|_2^{5p})
            &\text{if }\|r_{XX}\|_2\geq 1
            \\
            O_p(n^{-1/2}\tau_p^{-2})
            +O_p(n^{-1}\tau_p^{-3})
            \quad\text{(since $\tau_p\leq h_{jj}=O(1)$)}
            &\text{if }\|r_{XX}\|_2< 1.
        \end{cases}
        \label{eq:diff.H.alpha}
    \end{align}
	For each $(s,t)\in\mathbb I_X\times\mathbb I_Y$,
    \begin{align*}
        |\hat\beta_{p,\rm fAPLS}(s,t)&-\beta_{p,\rm fAPLS}(s,t)|^2
        \\
        =&\ \Bigg|
                [\widehat\Gamma_{XX}(\beta^*)(s,s'),\ldots,\widehat\Gamma_{XX}^p(\beta^*)(s,s')]
                \widehat{\bm H}_p^{-1}\widehat{\bm\alpha}_p
        \\      
                &-
                [\Gamma_{XX}(\beta^*)(s,s'),\ldots,\Gamma_{XX}^p(\beta^*)(s,s')]
                \bm H_p^{-1}\bm\alpha_p
            \Bigg|^2
        \\
        \leq&\ \Bigg|
            \|\widehat{\bm H}_p^{-1}\widehat{\bm\alpha}_p-\bm H_p^{-1}\bm\alpha_p\|_2
            \left[\sum_{j=1}^p\{\widehat\Gamma_{XX}^j(\beta^*)(s,s')\}^2\right]^{1/2}
        \\
            &+
            \|\bm H_p^{-1}\bm\alpha_p\|_2
            \left[\sum_{j=1}^p
                [\{\widehat\Gamma_{XX}^j-\Gamma_{XX}^j\}(\beta^*)(s,s')]^2\right]^{1/2}
        \Bigg|^2
        \\
        \leq&\ 
            2\|\widehat{\bm H}_p^{-1}\widehat{\bm\alpha}_p-\bm H_p^{-1}\bm\alpha_p\|_2^2
            \left[\sum_{j=1}^p\{\widehat\Gamma_{XX}^j(\beta^*)(s,s')\}^2\right]
        \\    
            &+
            2\|\bm H_p^{-1}\bm\alpha_p\|_2^2
            \left[\sum_{j=1}^p
                \{\widehat\Gamma_{XX}^j(\beta^*)(s,s')-\Gamma_{XX}^j(\beta^*)(s,s')\}^2\right].
    \end{align*}
    Thus $\|\hat{\beta}_{p,\rm fAPLS}-\beta_{p,\rm fAPLS}\|_2$ is bounded as below:
    \begin{align}
        \|&\hat{\beta}_{p,\rm fAPLS}-\beta_{p,\rm fAPLS}\|_2^2
        \notag\\
        \leq&\ 2\|\widehat{\bm H}_p^{-1}\widehat{\bm\alpha}_p-\bm H_p^{-1}\bm\alpha_p\|_2^2
            \sum_{j=1}^p\|\Gamma_{XX}^j(\beta^*)\|_2^2
            + 
            2\|\bm H_p^{-1}\aalpha_p\|_2^2
            \sum_{j=1}^p\|\Gamma_{XX}^j(\beta^*)-\widehat\Gamma_{XX}^j(\beta^*)\|_2^2
        \notag\\
        \leq&\ 2\|\widehat{\bm H}_p^{-1}\widehat{\bm\alpha}_p-\bm H_p^{-1}\bm\alpha_p\|_2^2
            \sum_{j=1}^p\|\Gamma_{XX}^j(\beta^*)\|_2^2
        \label{eq:beta.dist.1}\\
        &+ 2\tau_p^{-2}\|\bm\alpha_p\|_2^2
            \sum_{j=1}^p \|\widehat\Gamma_{XX}^j(\beta^*)-\Gamma_{XX}^j(\beta^*)\|_2^2,
        \label{eq:beta.dist.2}
    \end{align}
    where,
    owing to \eqref{eq:diff.H.alpha},
    $$
        \eqref{eq:beta.dist.1}
        =\begin{cases}
        	O_p(n^{-1}\tau_p^{-2}p^{4}\|r_{XX}\|_2^{4p})
        	\\
            \qquad +\ O_p(n^{-1}\tau_p^{-4}p^{6}\|r_{XX}\|_2^{8p})
            + O_p(n^{-2}\tau_p^{-6}p^{10}\|r_{XX}\|_2^{12p})
            &\text{if }\|r_{XX}\|_2\geq 1
            \\
            O_p(n^{-1}\tau_p^{-4})
            +O_p(n^{-2}\tau_p^{-6})
            &\text{if }\|r_{XX}\|_2< 1;
        \end{cases}
    $$
    the order of \eqref{eq:beta.dist.2} is jointly given by 
    \eqref{eq:bound.alpha.p} and \autoref{lemma:conv.gamma.i},
    i.e.,
    $$
        \eqref{eq:beta.dist.2}
        =\begin{cases}
			O(n^{-1}\tau_p^{-2}p^4\|r_{XX}\|_2^{4p})
            &\text{if }\|r_{XX}\|_2\geq 1
            \\
            O_p(n^{-1}\tau_p^{-2})
            &\text{if }\|r_{XX}\|_2< 1.
        \end{cases}
    $$
    In this way we deduce
    \begin{align}
       \|&\hat{\beta}_{p,\rm fAPLS}-\beta_{p,\rm fAPLS}\|_2^2
       \notag\\
       =&\ \begin{cases}
        	O_p(n^{-1}\tau_p^{-2}p^{4}\|r_{XX}\|_2^{4p})
        	\\
            \qquad +\ O_p(n^{-1}\tau_p^{-4}p^{6}\|r_{XX}\|_2^{8p})
            + O_p(n^{-2}\tau_p^{-6}p^{10}\|r_{XX}\|_2^{12p})
            &\text{if }\|r_{XX}\|_2\geq 1
            \\
            O_p(n^{-1}\tau_p^{-4})
            +O_p(n^{-2}\tau_p^{-6})
            &\text{if }\|r_{XX}\|_2< 1.
        \end{cases}
        \label{eq:beta.dist}
    \end{align}
    A set of necessary conditions for 
    the zero-convergence (in probability) of \eqref{eq:beta.dist}
    is contained in \ref{cond:diverge.p.add.L2}. 
    Once they are fulfilled, 
    we conclude the $L_2$ convergence (in probability) of $\hat{\beta}_{p,\rm fAPLS}$
    to $\beta^*$ following \autoref{prop:conv.beta.p}.
    
    We complete the proof by bounding the estimating error in the supremum metric:
    \begin{align*}
        \|&\hat{\beta}_{p,\rm fAPLS}-\beta_{p,\rm fAPLS}\|_\infty^2
        \\
        =&\ \left\|
                [\widehat\Gamma_{XX}(\beta^*),\ldots,\widehat\Gamma_{XX}^p(\beta^*)]
                \widehat{\bm H}_p^{-1}\widehat{\bm\alpha}_p
                -
                [\Gamma_{XX}(\beta^*),\ldots,\Gamma_{XX}^p(\beta^*)]
                \bm H_p^{-1}\bm\alpha_p
            \right\|_\infty^2
        \\
        \leq&\ 2\|\widehat{\bm H}_p^{-1}\widehat{\bm\alpha}_p-\bm H_p^{-1}\bm\alpha_p\|_2^2
            \sum_{j=1}^p\|\Gamma_{XX}^j(\beta^*)\|_\infty^2
            + 
            2\|\bm H_p^{-1}\aalpha_p\|_2^2
            \sum_{i=1}^p\|\Gamma_{XX}^j(\beta^*)-\widehat\Gamma_{XX}^j(\beta^*)\|_\infty^2
        \\
        \leq&\ 2\|\widehat{\bm H}_p^{-1}\widehat{\bm\alpha}_p-\bm H_p^{-1}\bm\alpha_p\|_2^2
            \sum_{j=1}^p\|\Gamma_{XX}^j(\beta^*)\|_\infty^2
        \quad\text{(compare \eqref{eq:beta.dist.1})}
        \\
        &+ 2\tau_p^{-2}\|\bm\alpha_p\|_2^2
            \sum_{j=1}^p \|\widehat\Gamma_{XX}^i(\beta^*)-\Gamma_{XX}^j(\beta^*)\|_\infty^2,
        \quad\text{(compare \eqref{eq:beta.dist.2})}
        \\
        =&\ \begin{cases}
        	O_p(n^{-1}\tau_p^{-2}p^{4}\|r_{XX}\|_\infty^{4p})
        	\\
            \qquad+\ O_p(n^{-1}\tau_p^{-4}p^{6}\|r_{XX}\|_\infty^{8p})
            + O_p(n^{-2}\tau_p^{-6}p^{10}\|r_{XX}\|_\infty^{12p})
            &\text{if }\|r_{XX}\|_\infty\geq 1
            \\
            O_p(n^{-1}\tau_p^{-4})
            +O_p(n^{-2}\tau_p^{-6})
            &\text{if }\|r_{XX}\|_\infty< 1,
        \end{cases}
    \end{align*}
    converging to zero (in probability) 
    with the satisfaction of \ref{cond:diverge.p.add.sup}.
    The zero-convergence (in probability) of 
    $\|\hat{\beta}_{p,\rm fAPLS}-\beta^*\|_\infty$ follows
    if we assume that $\|\beta_{p,\rm fAPLS}-\beta^*\|_\infty\to 0$ as $p$ diverges.
\end{proof}

\begin{proof}[Proof of \autoref{prop:conv.g.p.hat}]
    Notice that 
    \begin{align*}
        \|\hat g_{p,\rm fAPLS}(X_0)&-g(X_0)\|_2
        \\
        \leq&\ \|\bar Y-\mu_Y\|_2
			+\|\bar X-\mu_X\|_2\|\beta^*\|_2
			+\|X_0-\bar X\|_2\|\hat\beta_{p,\rm fAPLS}-\beta^*\|_2,
        \\
        \|\hat g_{p,\rm fAPLS}(X_0)&-g(X_0)\|_\infty
        \\
        \leq&\ \|\bar Y-\mu_Y\|_\infty
        	+\|\bar X-\mu_X\|_2\|\beta^*\|_\infty
            +\|X_0-\bar X\|_2\|\hat\beta_{p,\rm fAPLS}-\beta^*\|_\infty.
    \end{align*}
    The finite trace of $R_{XX}$ (resp. $R_{YY}$),
    viz. $\sum_{j=1}^\infty\lambda_{j,X}=\E(\|X-\mu_X\|_2^2)<\infty$ 
    (resp. $\sum_{j=1}^\infty\lambda_{j,Y}=\E(\|Y-\mu_Y\|_2^2)<\infty$),
    entails that
    $\|\bar X-\mu_X\|_2=o_{\rm a.s.}(1)$
    (resp. $\|\bar Y-\mu_Y\|_2=o_{\rm a.s.}(1)$); 
    see \citet[(2.1.3)]{Hoffmann-JorgensenPisier1976}.
    The proof is complete once we verify
    the zero-convergence (in probability and under \ref{cond:moment.add.sup}) 
    of $\|\bar Y-\mu_Y\|_\infty$
    following \citet[Theorem 2]{Hoffmann-Jorgensen1985}.
\end{proof}

\bibliographystyle{chicago}
\bibliography{mybibfile}
\end{document}